\documentclass[10pt, final, twocolumn, twoside,romanappendices]{IEEEtran}
\usepackage{winsnotation}
\usepackage[bookmarks,colorlinks]{hyperref} 
\hypersetup{colorlinks,citecolor= red,filecolor= blue,linkcolor= blue,urlcolor=blue}
\usepackage{algorithm}
\usepackage[noend]{algpseudocode}

\usepackage{tabularx,booktabs}
\newcolumntype{C}{>{\centering\arraybackslash}X} 
\usepackage{amsmath,amssymb}
\usepackage{mathtools}
\usepackage{graphicx}
\usepackage[font=small]{caption}
\usepackage[noadjust]{cite}
\usepackage{amsmath,amsfonts,amssymb,amsthm}

\usepackage{subfig}
\usepackage{mathtools}
\usepackage{commath}
\usepackage{acronym}  
\acrodef{ofdm}[OFDM]{Orthogonal Frequency-Division Multiplexing}
\acrodef{fft}[FFT]{fast Fourier transform}
\acrodef{clt}[CLT]{Central Limit Theorem}
\acrodef{uwb}[UWB]{ultra wideband}
\acrodef{iot}[IoT]{Internet of Things}
\acrodef{mimo}[MIMO]{multiple-input multiple-output}
\acrodef{isi}[ISI]{inter symbol interference}
\acrodef{cp}[CP]{cyclic prefix}
\acrodef{zp}[ZP]{zero-padded}
\acrodef{fir}[FIR]{finite impulse response}
\acrodef{v2x}[V2X]{Vehicle-to-everything}
\acrodef{nda}[NDA]{non-data-aided}
\acrodef{da}[DA]{data-aided}
\acrodef{ml}[ML]{maximum likelihood}
\acrodef{to}[TO]{timing offset}
\acrodef{wed}[WED]{weighted energy detector}
\acrodef{ed}[ED]{energy detector}
\acrodef{rms}[RMS]{root mean square}
\acrodef{ed}[ED]{energy detector}
\acrodef{tm}[TE]{transition estimator}
\acrodef{ge}[GE]{Gamma estimator}
\acrodef{pdf}[PDF]{probability density function}
\acrodef{cfo}[CFO]{carrier frequency offset}

\acrodef{iid}[i.i.d]{independent and identically distributed}

\acrodef{bem}[BEM]{basis expansion model}
\acrodef{ls}[LS]{least squares}
\acrodef{mmse}[MMSE]{minimum mean square error}
\acrodef{pa}[PA]{pilot-aided}
\acrodef{dd}[DD]{decision-directed}
\acrodef{cc}[CE]{channel estimation}
\acrodef{dnn}[DNN]{deep neural network}
\acrodef{mse}[MSE]{mean-squared error}
\acrodef{dl}[DL]{deep learning}
\acrodef{ci}[CI]{channel state information}
\acrodef{mmse}[MMSE]{minimum mean square error}
\acrodef{awgn}[AWGN]{additive white Gaussian noise}
\acrodef{map}[MAP]{maximum a posteriori probability}
\acrodef{ber}[BER]{bit error rate}
\acrodef{kf}[KF]{Kalman filter}
\acrodef{snr}[SNR]{signal-to-noise ratio}
\acrodef{iot}[IoT]{Internet of Things}

\acrodef{chf}[CHF]{characteristic function}

\newcommand{\bd}{\begin{description}}
\newcommand{\ed}{\end{description}}
\newcommand{\be}{\begin{enumerate}}
\newcommand{\ee}{\end{enumerate}}
\newcommand{\bi}{\begin{itemize}}
\newcommand{\ei}{\end{itemize}}
\newcommand{\bl}{\begin{list}}
\newcommand{\el}{\end{list}}
\newcommand{\bt}{\begin{tabbing}}
\newcommand{\et}{\end{tabbing}}

\usepackage{amsmath}
\usepackage{tikz}
\usepackage{mathdots}
\usepackage{yhmath}
\usepackage{cancel}
\usepackage{color}
\usepackage{siunitx}
\usepackage{array}
\usepackage{multirow}
\usepackage{amssymb}
\usepackage{gensymb}
\usepackage{tabularx}
\usepackage{booktabs}
\usetikzlibrary{fadings}
\usetikzlibrary{patterns}
\usetikzlibrary{shadows.blur}

\usepackage[yyyymmdd,hhmmss]{datetime} 
\newdateformat{monthyeardate}{\monthname[\THEMONTH] \THEDAY, \THEYEAR} 
\usepackage [autostyle, english = american]{csquotes}
\MakeOuterQuote{"}
\newtheorem{theorem}{Theorem}

\newtheorem{corollary}{Corollary}

\usepackage[usestackEOL]{stackengine}
\stackMath
\usepackage{stfloats}
\usepackage{arydshln}
\usepackage{booktabs}
\usepackage{siunitx}
\usepackage{caption}

\makeatletter
\let\old@ps@headings\ps@headings
\let\old@ps@IEEEtitlepagestyle\ps@IEEEtitlepagestyle
\def\confheader#1{%
  \def\ps@headings{%
    \old@ps@headings%
    \def\@oddhead{\strut\hfill#1\hfill\strut}%
    \def\@evenhead{\strut\hfill#1\hfill\strut}%
  }%
  \def\ps@IEEEtitlepagestyle{%
    \old@ps@IEEEtitlepagestyle%
    \def\@oddhead{\strut\hfill#1\hfill\strut}%
    \def\@evenhead{\strut\hfill#1\hfill\strut}%
  }%
  \ps@headings%
}
\makeatother

\confheader{%
  This Work was published by IEEE Transactions on Signal Processing in 2021.
}

\begin{document}
\title{Maximum Likelihood Time Synchronization for Zero-padded OFDM}
\author{
   Koosha~Pourtahmasi~Roshandeh,~\IEEEmembership{Student~Member,~IEEE},
   Mostafa~Mohammadkarimi,~\IEEEmembership{Member,~IEEE}, and
    Masoud~Ardakani,~\IEEEmembership{Senior~Member,~IEEE
    }
  \thanks{
        K.~P.~Roshandeh and M.~Ardakani are with the Faculty of Electrical and Computer Engineering, University of Alberta, Edmonton, AB T6G 2R3, Canada (e-mail: \texttt{pourtahm@ualberta.ca};   \texttt{ardakani@ualberta.ca)}.
        
         M.~Mohammadkarimi is with the
        Department of Electrical and Computer Engineering,
        The University of British Columbia, Vancouver, BC, V6T 1Z4, Canada,
        (e-mail: \texttt{mmkarimi@ece.ubc.ca})
	}
}

\maketitle

\begin{abstract}
Existing \ac{ofdm} variants based on cyclic
prefix (CP) allow for efficient time synchronization, but  suffer from lower power efficiency compared to \ac{zp}-\ac{ofdm}.  Because of its power efficiency, \ac{zp}-\ac{ofdm} is considered as an appealing solution for the emerging low-power
wireless systems. However, in the absence of CP, time synchronization in ZP-OFDM is a very challenging task. In
this paper, the non-data-aided (NDA) maximum-likelihood (ML)
time synchronization for ZP-OFDM is analytically derived. We
show that the optimal NDA-ML synchronization algorithm offers a 
high lock-in probability and can be efficiently implemented using  Monte Carlo sampling (MCS) technique in combination with
golden-section search.  To obtain the optimal NDA-ML time synchronization algorithm,
we first derive a closed-form expression for the joint probability
density function (PDF) of the received ZP-OFDM samples in
frequency-selective fading channels. The derived expression is
valid for doubly-selective fading channels with mobile users as
well. The performance of the proposed synchronization
algorithm is evaluated under various practical settings through
simulation experiments. It is shown that the proposed optimal 
NDA-ML synchronization algorithm and its MCS implementation substantially
outperforms existing algorithms in terms of lock-in probability.

\end{abstract}

\acresetall	
\section{Introduction}
\IEEEPARstart{O}{}rthogonal  frequency-division multiplexing (OFDM) modulation is a widely used technique for  transmission over mobile wireless channels since it 
offers high spectral efficiency whilst providing resilience to frequency-selective fading \cite{farhang2016ofdm}.
One of the key requirements for optimum demodulation of OFDM signals  is accurate time  synchronization
because a small synchronization error can dramatically degrade the system performance. Hence, a variety of time synchronization methods
have been developed for OFDM systems \cite{zhang2019fine,abdzadeh2019timing, zhang2015maximum, mohebbi2014novel,morelli2007synchronization,  park2004blind}. These methods typically consist of two tasks: 1) offset estimation, and 2) offset correction. 
The former task relies on statistical signal processing algorithms to obtain an estimation of \ac{to} incurred due to lack of common time reference between the transmitter and receiver \cite{lin2018analysis,ziamaxli2018, abdzadeh2016improved}. 
The latter task is a simple compensation of \ac{to} by shifting \cite{morelli2007synchronization}.

Time synchronization for OFDM can be performed using either synchronization-assisting  signals, such as pilot signals and synchronization symbols \cite{gul2014timing}, or exploiting some redundant information in the transmitted signal, such as the guard interval redundancy employed to combat the \ac{isi} in frequency-selective
fading channels. The former approach is  \ac{da}, and the latter is \ac{nda} time synchronization \cite{nasir2016timing}. \ac{da} time  synchronization comes at the cost of reduced spectral efficiency, especially for short burst transmission, which is widely
employed in \ac{iot} use case of the fifth generation (5G) wireless systems \cite{de20195g}.

Guard intervals are useful   for time synchronization in OFDM systems  \cite{ven5G}. The guard interval can be in the form of \ac{cp} \cite{wang2015maximum,chin2011blind,van1997ml}, \ac{zp} \cite{wang2011frequency,  su2008new, wang2006frames}, and known symbol padding (KSP) \cite{van2012iterative}. The choice of \ac{zp} versus \ac{cp}  and KSP depends on several parameters, such as the operating \ac{snr}, delay spread of the fading channel, and coherent versus differential demodulation.
\ac{zp}-OFDM provides great benefits over CP-OFDM and KSP-OFDM in the sense that \cite{giannazpcp} 1) it guarantees symbol recovery regardless of the channel zero locations; hence, it can improve the BER, 2) it enables finite impulse response equalization of channels regardless of the channel nulls, 3) it makes channel estimation and channel tracking easier compared to that of CP-OFDM, and 4) it offers higher power efficiency.

While \ac{da} time synchronization for \ac{zp}-OFDM has been well explored in the literature \cite{nasir2016timing}, \ac{nda} approach has not been extensively investigated. Hence, the focus of this work is on \ac{nda} time synchronization for \ac{zp}-OFDM. 
\subsection{Related Work}
For \ac{zp}-OFDM, most existing \ac{da} approaches rely on periodic autocorrelation properties of the received signal induced by the employed training sequences with good autocorrelation properties \cite{li2008synchronization}. 
Moreover, most of the \ac{da} approaches developed for \ac{cp}-OFDM can be applied to \ac{zp}-OFDM \cite{chung2017preamble,zhang2011autocorrelation, abdzadeh2012novel,sheng2010novel}.  
On the other hand, 
to the best of the authors’ knowledge, the few existing \ac{nda} synchronization approaches for \ac{zp}-OFDM have been developed based on change point detection methods \cite{LeNir2010} or cyclostationarity properties in OFDM signal \cite{bolcskei2001blind}. Synchronization algorithms based on change point detection usually employ a transition metric, tracing the ratio of power in two slicing windows corresponding to each \ac{to} hypothesis in the OFDM packet.
 These \ac{nda} solutions do not always offer a good performance in terms of lock-in probability, i.e. correct synchronization. Moreover, we show that their performance further drops in doubly-selective (time- and frequency-selective) fading channels. 
 \vspace{-0.5em}
\subsection{Motivation}
In the presence of perfect time synchronization, \ac{zp}-OFDM offers higher reliability and power efficiency compared to \ac{cp}-OFDM \cite{giannazpcp}. 
Hence, \ac{zp}-OFDM  can be considered  an appealing solution for low-power \ac{iot} networks. 
One of the main reasons that \ac{zp}-OFDM has not been extensively used 
 in practice is attributed to the lack of an efficient time synchronization method.  
While sub-optimal \ac{cp}-based synchronization algorithms in \ac{cp}-OFDM offer high lock-in  probabilities,
there is no  synchronization algorithm with comparable performance for \ac{zp}-OFDM \cite{nasir2016timing}.
In addition, derivation of the optimal \ac{ml} time synchronization for \ac{zp}-OFDM, which results in a high
lock-in  probability,
has remained intact. This is mainly because there currently exist no compact expression for the joint  \ac{pdf} of the received samples. 
Moreover, most existing synchronization methods ignore time-selectivity of the fading channel, i.e., the devastating effect of mobility and Doppler spread on time synchronization.    

Motivated by the advantages of \ac{zp}-OFDM for the emerging low-power wireless networks, we study the problem of  \ac{nda}-\ac{ml} time synchronization for \ac{zp}-OFDM. 
In the first step, and for the first time, we derive a closed-form expression for the \ac{pdf} of the received \ac{zp}-OFDM samples in frequency-selective fading channel. We then use the PDF of the samples to approximate their joint PDF. The joint \ac{pdf} is given to a hypothesis testing algorithm to find the \ac{to}. Simulation results show that the proposed NDA-ML time synchronization algorithm significantly outperforms other existing \ac{nda} time synchronization methods. For example, at 5 dB $E_{\rm{b}}/N_0$ for WiMAX SUI-4 channels \cite{LeNir2010}, the
proposed \ac{nda}-ML time synchronization algorithm achieves a lock-in probability of 0.85 while the state of the art \cite{LeNir2010} achieves 0.55. 
 \vspace{-0.5em}
\subsection{Contributions}
The main contributions of this paper are as follows: 

\begin{itemize}

\item A closed-form approximate expression for the joint \ac{pdf} of the received \ac{zp}-OFDM samples in frequency-selective fading channels is derived.

  \item The \ac{nda}-\ac{ml} time synchronization for \ac{zp}-OFDM  in frequency-selective fading channels  is analytically derived. 
    The proposed method exhibits the following advantages: (i) unlike existing sub-optimal \ac{nda} time synchronization methods, it is applicable to highly selective fading channels, such as the ones in underwater communications and \ac{uwb} communication, (ii) it is valid for doubly-selective fading channels, and (iii) it can be used for both frame and symbol synchronization.


\item A  low-complexity implementation of the developed theoretical \ac{nda}-\ac{ml} time synchronization algorithm by using Monte Carlo sampling (MCS) technique and  golden-section  search is proposed.

\item Complexity analysis of the proposed time synchronization methods is provided.

\end{itemize}

The remaining of the paper is organized as follows: Section~\ref{sec: sys model} introduces the system model. Section~\ref{sec: ml estimator} describes the derivation of the \ac{nda}-\ac{ml} time synchronization. In Section~\ref{sec: importance samp},  a practical implementation of the proposed time synchronization algorithm by employing MCS technique and golden-section  search is presented. Simulation results are provided in Section~\ref{simmp}, and conclusions are drawn in Section~\ref{sec: conclu}.

\textit{Notations}: Throughout this paper, we use bold lowercase and 
bold uppercase letters to show  
column vectors and matrices, respectively. The symbols  $(\cdot)^{*}$, $(\cdot)^{\{rm{T}}$, $|\cdot|$, and $\lfloor{\cdot}\rfloor$ denote conjugate, transpose, absolute value, and
floor function, respectively. $\mathbb{E}\{\cdot\}$ denotes 
the statistical expectation, and $\Re\{\cdot\}$ and $\Im\{\cdot\}$ represent the the real and
imaginary parts, respectively, The subscripts 
${\rm{I}}$ and ${\rm{Q}}$ show the in-phase and quadrature components of a variable. The symbols $\bigcap$ and $\bigcup$ denote the set intersection and union operands, respectively.

\section{System Model} \label{sec: sys model}
We consider a ZP-OFDM system in frequency-selective fading channel. Let $\{x_{n,k}\}_{k=0}^{n_{\rm{x}}-1}$, $\mathbb{E}\{|x_{n,k}|^2\}= \sigma^2_{\rm{x}}$, be the $n_{\rm{x}}$ complex data to  be transmitted  in the $n$-th  OFDM  symbol. The OFDM modulated baseband signal is given by \cite{li2006orthogonal,hwang2008ofdm}
\begin{align}\label{477www131314113}
x_n(t)=\sum_{k=0}^{n_{{\rm{x}}-1}}x_{n,k} e^{\frac{j2\pi k t}{T_{\rm{x}}}},\,\,\,\,\,\,\,\,\,\,\ 0 \le t \le T_{\rm{x}},
\end{align}
where  $T_{\rm{x}}$ and $W\triangleq n_{\rm{x}}/T_{\rm{x}}$
is the OFDM symbol duration and channel bandwidth, respectively.
To avoid \ac{isi}, zero-padding guard interval of length $T_{\rm{z}}$ is added to each OFDM symbol. Hence, $x_n(t)$ is extended into $s_n(t)$ as
\begin{align}\label{cp_add}
s_n(t)=
\begin{cases}
x_n(t)  \,\,\,\,\,\,\,\,\,\,\,\ 0 \le t < T_{\rm{x}}  \\
0 \,\,\,\,\,\,\,\,\,\,\,\,\,\,\,\,\,\,\,\,\,\,\,\,  T_{\rm{x}} \le t < T_{\rm{x}}+T_{\rm{z}}.
\end{cases}
\end{align}
The OFDM signal in \eqref{cp_add} propagates through a multi-path fading channel with the equivalent baseband  impulse response as follows
\begin{align}
h(\tau)=\sum\limits_i {}\alpha_i \delta (\tau-\tau_i),
\end{align}
where $\alpha_i \in \mathbb{C}$. The delay spread of the channel in the ensemble sense is $\tau_{\rm{d}}$ where 
$\mathbb{E}\{|\alpha_i|^2\}=0$ for $\tau_i > \tau_{\rm{d}}$.

When the transmitter and receiver are synchronized and there is no \ac{isi}, i.e. $T_{\rm{z}} \ge \tau_{\rm{d}}$,
the complex baseband received signal sampled at multiples of $T_{\rm{sa}} \triangleq  1/W$  is given by 
{
\begin{align}\label{uiomoer21}
y_{n}[k]=\sum_{l=0}^{n_{\rm{h}}-1} h[nn_{\rm{s}}+k;l] s_{n}[k-l]+w_n[k],
\end{align}}
$m=0,1,\dots,n_{\rm{s}}-1$, where $n_{\rm{s}}=\lfloor(T_{\rm{x}}+T_{\rm{z}})/T_{\rm{sa}}\rfloor$, $s_n[m]\triangleq s_n(mT_{\rm{sa}})$,
\begin{align}
h[l]=\sum\limits_i {} \alpha_i 
g[l-\tau_iW],
\end{align}
$l=0,1,\dots,n_{\rm{h}}-1$, 
$n_{\rm{h}} \triangleq \lfloor\tau_{\rm{d}}/T_{\rm{sa}}\rfloor$, $g[l] \triangleq  g(lT_{\rm{sa}})$,
$g(t) \triangleq g_{\rm{Tx}}(t)\circledast g_{\rm{Rx}}(r)$ with 
$g_{\rm{Tx}}(t)$ and $g_{\rm{Rx}}(t)$ as transmit and receive filters, respectively. Also, $w_n[m] \sim \mathcal{CN}(0,\sigma^2_{\rm{w}})$ is the \ac{awgn}.

We consider the wide-sense stationary uncorrelated scattering (WSSUS) assumption so that the channel coefficients from different delay taps are independent. The channel taps $h[l]$, $l=0,1,\dots,n_{\rm{h}}-1$, are modeled as  statistically independent zero-mean complex Gaussian  random  variables (Rayleigh  fading) with the delay profile   
\begin{align}\label{7u8i0000}
\mathbb{E}{\{}h[l]h^*[l-m]{\}}=\sigma_{{{\rm{h}}_l}}^2\delta[m],
\end{align}
$l=0,1,\dots, n_{\rm{h}}-1$, where 
\begin{align}
\sigma_{{{\rm{h}}_l}}^2=\mathbb{E}\{|h[l]|^2\}=\sum_i \mathbb{E}{\{}|\alpha_i|^2\}{|}g[l-\tau_iW]{|}^2. 
\end{align}
It is assumed that the delay profile of the fading channel is known to the receiver.

{\it Remark 1:} The PDP of an environment is obtained through field measurements by transmitting a short pulse (wide-band) and measuring the received power as a function of delay at various locations in a small area during channel sounding.
These measurements are then averaged over spatial locations to generate a profile of the average received signal power as a function of delay.
\cite{Delay}.
 Theoretically, the PDP is defined as the expectation of the squared impulse response of the channel as 
\begin{equation}
\breve{p}(\tau) = \mathbb{E}\{|h(\tau)|^2\}.
\end{equation}
Assuming WSSUS scattering, the PDP is given as
\begin{equation}
\breve{p}(\tau) = \sum_{k=0}^{N}\alpha_k \delta(\tau-\tau_k).
\end{equation}
To determine the number of paths $N$, different criteria for model order selection are available in the existing literature. Estimating  the  path  delays  using  frequency domain pilots is equivalent to estimating the arrival angle using an  antenna  array  \cite{swindlehurst1998time}. Hence,  well-known  signal  processing techniques, e.g., estimation of signal parameters via rotational in-variance techniques (ESPRIT) \cite{roy1989esprit},  can be adopted for this purpose. With the estimates of path delays, the path gains $\alpha_k$, $k=0,1,\dots,N$, can be obtained using typical linear estimators \cite{liu2014channel}.

We define 
$n_{\rm{z}} \triangleq \lfloor{T_{\rm{z}}}/{T_{\rm{sa}}}\rfloor$
as the number of padded zeros. Hence, the number of samples per \ac{zp}-OFDM  symbol is
$n_{\rm{s}} \triangleq n_{\rm{x}}+n_{\rm{z}}$.
Equation \eqref{uiomoer21} can be written in a vector form as follows
\begin{align}\label{signal}
{\bf y}_{n}=
\begin{cases}
 {\bf H} {\bf s}_n   +{\bf w}_n \triangleq  {\bf v}_n +{\bf w}_n,  \,\,\,\,\,\,\,\,\,\,\,\ n \ge 0  \\
{\bf w}_n, \,\,\,\,\,\,\,\,\,\,\,\,\,\,\,\,\,\,\,\,\,\,\,\,\,\,\,\,\,\,\,\,\,\,\,\,\,\,\,\,\,\,\,\,\,\,\,\,\,\,\,\,\,\,\,\,\,\,\,\,\,\, n<0,
\end{cases}
\end{align}
where 
\begin{equation}\label{eq:11}
{\bf{s}}_n \triangleq
\left[ \begin{array}{l}
\,\,\,\,\,\ s_n[0]\\
\,\,\,\,\,\ s_n[1]\\
\,\,\,\,\,\,\,\,\,\ \vdots \\
s_n[{n_{\rm{s}}}-1]
\end{array} \right]=
\left( 
                                  \begin{array}{c}
                       x_{n}(0)\\
                                         \vdots\\
                                         x_n((n_{\rm x}-1)T_{\rm{sa}}) \\
                                               0\\
                                         \vdots\\
                                         0\\
          
                                  \end{array}
                            \right)
\setstackgap{L}{1.2\normalbaselineskip}
\vcenter{\hbox{\stackunder[1pt]{%
  \left.{\Centerstack{\\ \\}}\right\}n_{\rm{x}}%
}{
  \left.{\Centerstack{\\ \\}}\right\}n_{\rm{z}}%
}}},
\end{equation}
\begin{subequations} \label{uuoio}
   \begin{align}
    {\bf y}_{n} &\triangleq \Big{[}y_n[0]~y_n[1]~\dots~  y_n[n_{\rm{s}}-1] \Big{]}^{\rm{T}},\\ \label{eq:11a}
    {\bf w}_n &\triangleq \Big{[}w_n[0]~w_n[1]~\dots~  w_n[n_{\rm{s}}-1] \Big{]}^{\rm{T}},\\
    {\bf v}_n &\triangleq \Big{[}v_n[0]~ v_n([1] \dots  v_n[n_{\rm{s}}-1] \Big{]}^{\rm{T}},
\end{align}
\end{subequations}

${\bf H}$ is an $n_{\rm{s}}\times n_{\rm{s}}$ matrix, where its $i$-th ($0\le i \le n_{\rm{s}}-1$) column is $[ {\bf 0}_{i-1} \ h[nn_{\rm{s}}+i-1;0] \ h[nn_{\rm{s}}+i-1;1] \ \dots \ h[nn_{\rm{s}}+i-1;n_{\rm h}-1] ~ {\bf{0}}_{n_{\rm{s}}-n_{\rm h}-i+1}]^\text{T}$, ${\bf v}_n \triangleq{\bf H} {\bf s}_n$, and ${\bf w}_n$ is the \ac{awgn} vector.

Based on the \ac{clt}, the $T_{\rm{sa}}$-spaced baseband OFDM samples can be accurately modeled by \ac{iid} zero-mean complex Gaussian random variables as follows \cite{banelli2003theoretical}
\begin{align} \label{909m0p}
x_n(mT_{\rm{sa}}) \sim \mathcal{CN}(0,\sigma^2_{\rm{x}}), 
\end{align}
where 
\begin{align}
\mathbb{E}\Big{\{}x_n(mT_{\rm{sa}})x_n^*(kT_{\rm{sa}})\Big{\}}=\sigma^2_{\rm{x}} \delta[m-k].   
\end{align}

We consider that the transmitter and receiver are not synchronized in time domain, and there is a \ac{to} between them defined as $\tau \triangleq dT_{\rm{sa}}+\epsilon$, where  $d$ and $\epsilon$ represent the integer and fractional part of the \ac{to}. The fractional part of the delay appears as phase offset at each  sub-carrier. Hence, its effect is compensated when carrier frequency offset  is estimated \cite{morelli2007synchronization}. However, estimation of the integer part $d$ is required in order to detect the starting point of the \ac{fft} at the receiver. Estimating the integer part of the \ac{to} is the subject of this paper. We consider that the transmitter does not use pilot or preamble for \ac{to} estimation; thus, the receiver relies on the received samples, noise samples in the zero-guard interval, and the second-order statistics of the fading channel to estimate the \ac{to}.

\section{Maximum Likelihood Estimation}
\label{sec: ml estimator} 
In this section, we analytically derive the \ac{nda}-ML \ac{to} estimator for ZP-OFDM. 
For the ease of discussions and presentation, we consider $d\in \{-n_{\rm{s}}+1,\dots,-1,0,1,\dots,n_{\rm{s}}-1\}$. However,
the range of $d$  can be considered arbitrary large.


We first formulate \ac{to} estimation as a 
multiple hypothesis testing problem as 
${\rm{H}}_{p}: d=p$ where $-n_{\rm{s}}+1 \le p \le n_{\rm{s}}-1$. Since both positive and negative values of \ac{to} are considered, the ML estimator can address frame and OFDM symbol synchronization. 
Considering negative \ac{to} enables us to find the onset of the packet and incorporating positive \ac{to} enables us to find the starting point of the  ZP-OFDM symbols to efficiently apply \ac{fft} for channel equalization and data detection.

We consider that the OFDM receiver gathers 
$N$ observation vectors of length $n_{\rm{s}}$, ${\bf y}_{0}, {\bf y}_{1}, \dots, {\bf y}_{N-1}$, to estimate
the \ac{to}, $d$. The initial step for \ac{ml} derivation is to obtain the joint \ac{pdf} of the observation vectors under the $2n_{\rm{s}}+1$ \ac{to} hypotheses. We denote this joint \ac{pdf} given ${\rm{H}}_d$ by $f_{\rm{Y}}({\bf y}| {\rm{H}}_{d})$, where 
\begin{align}\label{eq:14}
 {\bf{y}} &= \big{[}y[0] \ y[1] \ \dots \ y[Nn_{\rm{s}}-1]\big{]}^{\rm{T}} \\ \nonumber
 &\triangleq[{\bf{y}}_0^{\rm{T}} \  {\bf{y}}_1^{\rm{T}} \  \dots \ {\bf{y}}_{N-1}^T {]}^{\rm{T}}   
\end{align}
with 
\begin{align}\label{eq:61}
y[nn_{\rm{s}}+m]  \triangleq y_{n}[m].
\end{align} 
as the $m$-th sample in the $n$-th block. By using the chain rule in probability theory \cite{leon1994probability}, we can write
\begin{align}\label{bnbnm}
f_{\bf Y}({\bf y}| &{\rm{H}}_{d}) \\ \nonumber 
  &= \prod_{n=0}^{N-1}  \prod_{m=0}^{n_{\rm{s}}-1}  f\Bigg{(}y_n[m] \,\Big|\,  \bigcap_{u=0}^{m-1} y_n[u], \bigcap_{k=0}^{n-1}{\bf{y}}_k, {\rm{H}}_d\Bigg{)}.
\end{align}
To obtain the joint \ac{pdf} in \eqref{bnbnm},
 we rely on Theorem \ref{uuonkml}.

\begin{theorem}\label{uuonkml}
The elements of the observation vector $\bf{y}$  in \eqref{eq:14}  irrespective to the value of \ac{to} are uncorrelated random variables, i.e., $\mathbb{E}\{y_n[u]y_{\tilde{n}}^*[v]\}=0$, $u \neq {{v}}$.
\end{theorem}

\begin{proof} See Appendix \ref{proof: theo uncorrelated}.
\end{proof}

Fig. \ref{fig:example} illustrates the scatter plot of the in-phase components of ${y_{10}[100]}$ (i.e., ${y_{{10}_{\rm{I}}}[100]}=\Re\{{y_{10}[100]}\}$) and ${y_{10}[101]}$ (i.e., ${y_{{10}_{\rm{I}}}[101]}=\Re\{{y_{10}[101]}\}$)  given hypothesis ${\rm{H}}_0$. As seen, there is no correlation between the two successive samples.


\begin{figure}
\vspace{-1em}
\centering
    \includegraphics[height=2.835in]{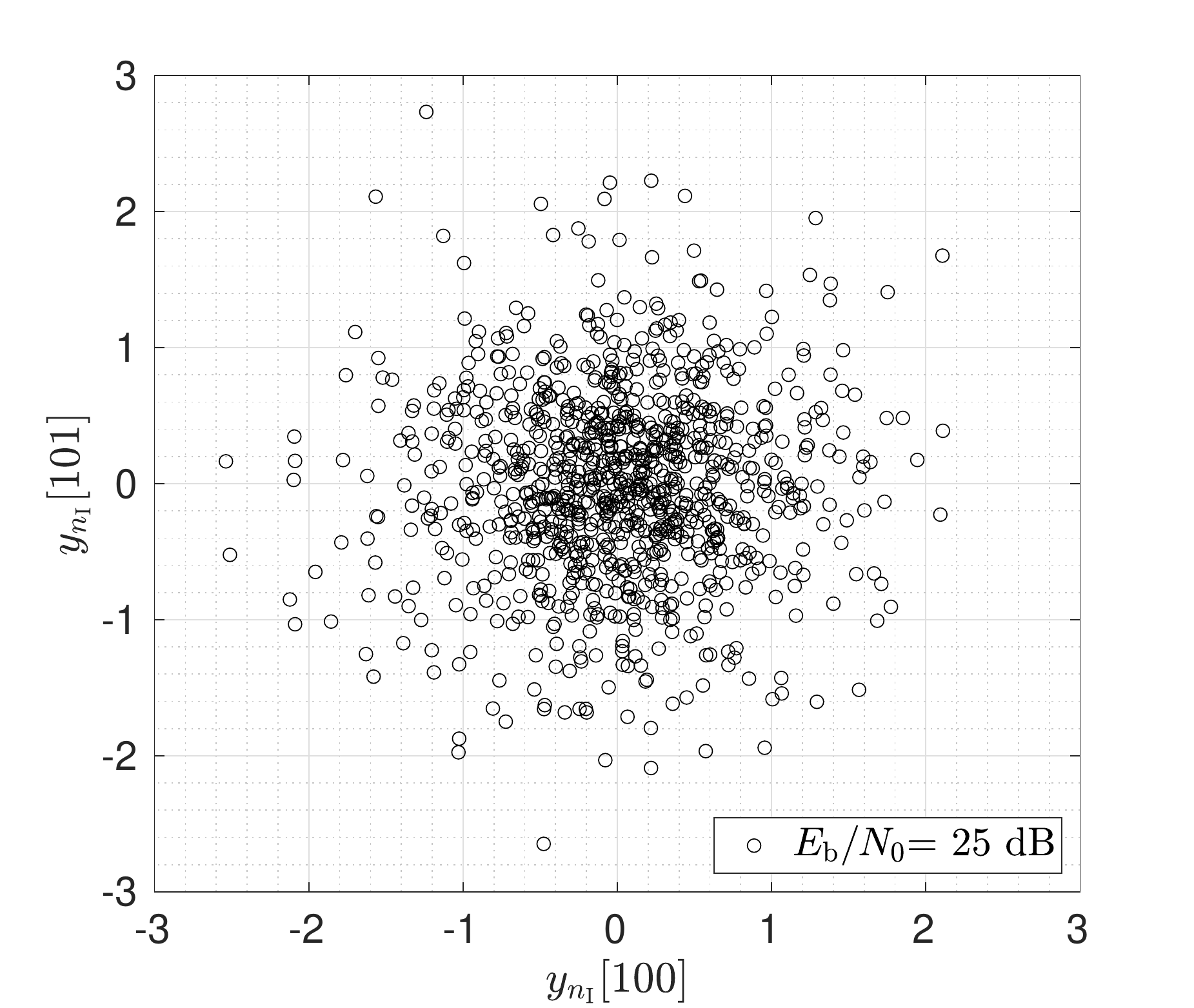}  
    \caption{Scatter plot of $y_{10_{\rm{I}}}[100]$ and $y_{10_{\rm{I}}}[101]$ given hypothesis ${\rm{H}}_0$ at $25$ dB $E_{\rm{b}}/N_0$  ($n_{\rm{x}}=128, n_{\rm{z}}=15, n_{\rm{h}}=10, n=10$). }
    \label{fig:example}%
\label{fig: correlation}
\end{figure}
\begin{figure}
\vspace{-2em}
\centering
\includegraphics[height=2.835in]{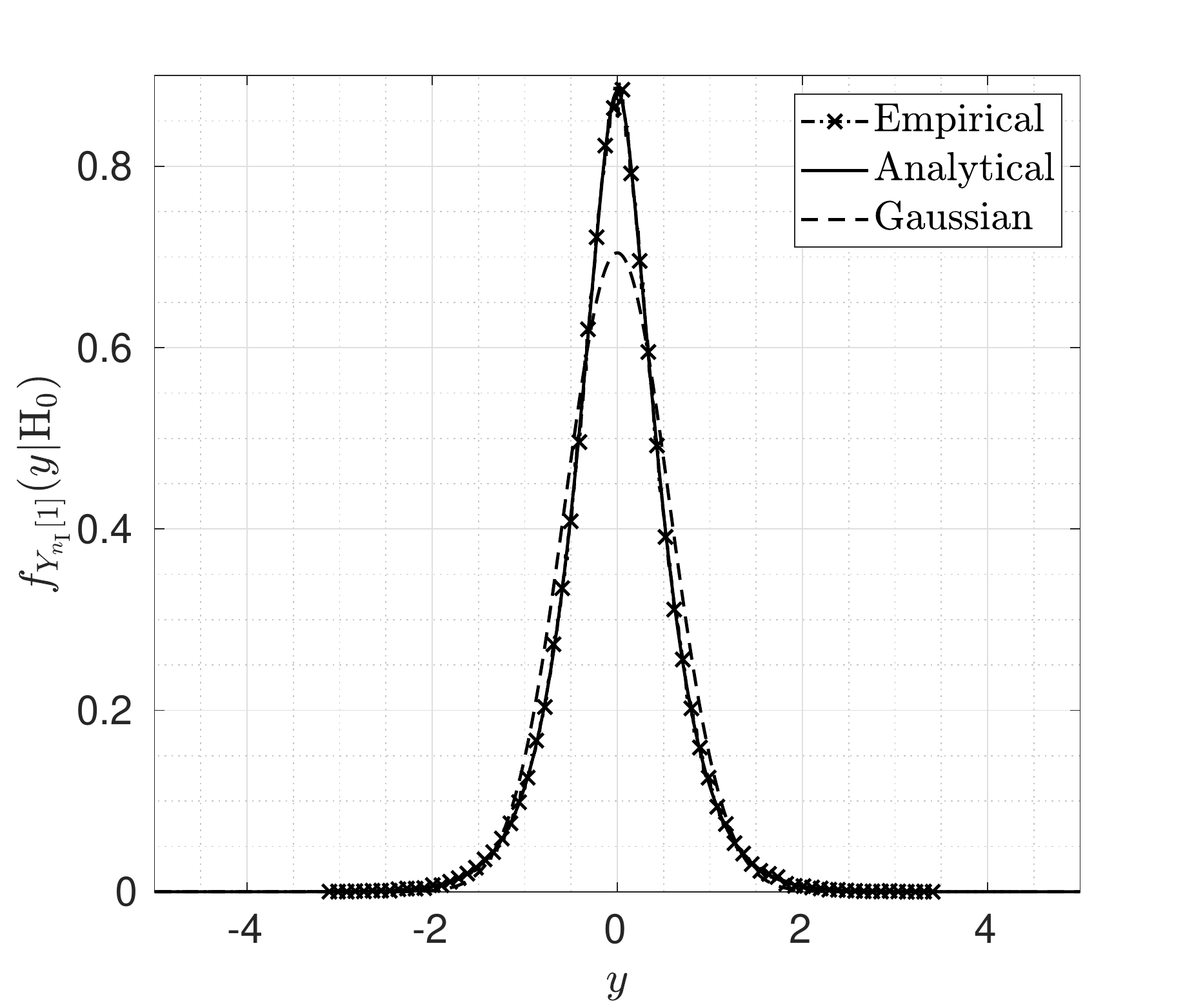}
  \caption{The empirical and analytical \ac{pdf}s of ${y_{n_{\rm{I}}}[1]}$ given hypothesis ${\rm{H}}_0$ at $15$ dB $E_{\rm{b}}/N_0$ ($n_{\rm{x}}=128, n_{\rm{z}}=15, n_{\rm{h}}=10, n=10$). The empirical \ac{pdf} was obtained for $10^6$ samples.}\label{fig: Empirical vs Analytical}
  \vspace{-1em}
\end{figure}

According to Theorem \ref{uuonkml}, the observation samples in \eqref{eq:14} are uncorrelated random variables. Also, we can show that the in-phase $y_{n_{{\rm{I}}}}[m]$ and quadrature $y_{n_{{\rm{Q}}}}[m]$ components of the $m$-th received sample from the $n$-th observation vector, i.e., 
$y_n[m]=y_{n_{{\rm{I}}}}[m]+
iy_{n_{{\rm{Q}}}}[m]$, 
are uncorrelated random variables.

Although uncorrelated random variables are not necessarily independent, the independency assumption  becomes more valid for the received samples in the case of very fast-varying channels where 
the maximum Doppler spread of the channel approaches infinity. When the maximum Doppler spread reaches infinity, the channel taps contributing to one received sample become independent from other samples.
Moreover, most practical algorithms are particularly sensitive to the distribution and less to correlation. Thus, we can consider that the observation samples are independent random variables to simplify the signal model \cite{kay2013fundamentals}.  Accordingly, we approximate  
the joint \ac{pdf} in \eqref{bnbnm} by the multiplication of the first-order \ac{pdf}s as
\begin{align}\label{eq: indepency 1}
\hspace{-0.2em} f_{\bf{Y}}({\bf y}| {\rm{H}}_{d})  &\approx 
    \prod_{n=0}^{N-1}
\prod_{m=0}^{n_{\rm{s}}-1}  f_{Y_n[m]}(y_n[m]| {\rm{H}}_{d})\\ \nonumber 
& \approx\prod_{n=0}^{N-1}
\prod_{m=0}^{n_{\rm{s}}-1}  f_{Y_{n_{\rm{I}}[m]}}(y_{n_{{\rm{I}}}}[m] |  {\rm{H}}_{d} ) f_{Y_{n_{\rm{Q}}}[m]}(y_{n_{{\rm{Q}}}}[m] |  {\rm{H}}_{d}),
\end{align}
where $f_{Y_{n_{\rm{I}}}[m]}(\cdot |  {\rm{H}}_{d}) $ and $f_{Y_{n_{\rm{Q}}}[m]}(\cdot |  {\rm{H}}_{d})$ are the \ac{pdf} of the in-phase and quadrature components of the $m$-th received sample from the $n$-th observation vector.  
The accuracy of this assumption is verified in the simulation results.

\subsection{\ac{pdf} Derivation for Delay Hypothesis ${\rm H}_0$ }

Theorem \eqref{theo: pdf y} provides closed-form expressions for the \ac{pdf} of the in-phase and quadrature components of the  $m$-th received sample from the $n$-th observation vector given the hypothesis $\mathrm{H}_0$ ($d=0$). 
We later show that the conditional \ac{pdf}s given hypothesis $\mathrm{H}_d$, $d \ne 0$, can be easily extracted from these \ac{pdf}s due to
the periodicity incurred by the zero-padded guard interval.

\begin{theorem} \label{theo: pdf y}
The PDF of the in-phase (quadrature) component of the received samples  $y_{n_{{\rm{I}}}}[m]$ ($y_{n_{{\rm{Q}}}}[m]$), $m \in  \{ m ~|~ 0 \le m \le n_{\rm{s}}-1  ~{\rm{when}}~  n<0 \} \cup \{ m ~|~ n_{\rm{x}}+n_{\rm{h}}-1 \le m \le n_{\rm{s}}-1 ~{\rm{when}}~  n \ge 0\}$ for $d=0$     
is given by  \footnote {In Theorem \ref{theo: pdf y}, by $n<0$, we mean as if the receiver starts to receive samples before any data is transmitted from the transmitter. 
} 
\begin{equation} \label{uioooppioio}
\begin{split}
f_{Y_{n_{\rm{I}}}[m]}(y|{\rm{H}}_0) &=
\frac{1}{\sqrt{\pi \sigma^2_{\rm w} } }  \exp\Big{(}-\frac{y^2}{\sigma^{2}_{\rm w}} \Big{)}.
\end{split}
\end{equation}
Also, for $d=0$ and $m \in  \{ m ~|~ 0 \le m \le n_{{\rm{x}}}+n_{{\rm{h}}}-2~{\rm{when}}~  n \ge 0\}$, we have
\begin{align} \label{eq: pdf}
f_{Y_{n_{\rm{I}}}[m]}& (y|{\rm{H}}_0)=f_{Y_{n_{\rm{I}}}[m]} (-y|{\rm{H}}_0) \\ \nonumber
&=
\Big{(}\prod_{i=a}^{b} \lambda_i\Big{)}^2 \sum_{j=a}^{b}  \sum_{n=a}^{b} \frac{ e^{\big{(}\frac{\lambda_j \sigma_{\rm{w}}}{2}\big{)}^2}}{ \prod_{k=a, k \neq j}^{b} (\lambda_k - \lambda_j)} \\ \nonumber
&~ \times \frac{ 1 }{ \prod_{u=a, u \neq j}^{b} (\lambda_u - \lambda_n)} \frac{1}{(\lambda_j + \lambda_n) } \\ \nonumber 
& \times 1/2 \Bigg[ e^{-\lambda_j y} \Bigg( 1-\Phi \Big(\frac{\lambda_j \sigma_{\rm{w}}}{2}-\frac{y}{\sigma_{\rm{w}}}\Big) \Bigg) + \\
&~~~~~~~~~~~~~~~~~~~~~~ e^{\lambda_j y} \Bigg(1-\Phi \Big(\frac{\lambda_j \sigma_{\rm{w}}}{2}+\frac{y}{\sigma_{\rm{w}}} \Big) \Bigg)  \Bigg],
\end{align}
\noindent where  $\lambda_k \triangleq 2/( \sigma_{h_k}  \sigma_{{\rm x}})$, $\Phi(x)= {\rm{erf}}(x)=\frac{2}{\sqrt{\pi}}\int_{0}^{x} e^{-t^2} dt$ denotes the Gaussian error function, 
and $a$ and $b$ depend on $m$, and are given as follows
\begin{equation} \label{eq: a b no isi}
\hspace{-1em}
(a,b) = 
     \begin{cases}
      (0,m) &~ 0 \le m \le n_{\rm{h}}-2\\
       (0,n_{\rm{h}}-1) &~ n_{\rm{h}}-1 \le m \le n_{\rm{x}}-1\\
       (m-n_{\rm{x}}+1,n_{\rm{h}}-1) &~ n_{\rm{x}} \le m \le n_{\rm{x}}+n_{\rm{h}}-2.\\
     \end{cases}
\end{equation}
Similar expressions hold for $f_{Y_{n_{\rm{Q}}}[m]}(y|{\rm{H}}_0)$.
\end{theorem}
\begin{proof}
See Appendix \ref{proof: theo}.
\end{proof}

Fig. \ref{fig: Empirical vs Analytical} illustrates the derived \ac{pdf} in \eqref{eq: pdf} for ${y_{n_{\rm{I}}}[1]}$ given hypothesis ${\rm{H}}_0$ at 15 dB $E_{\rm{b}}/N_0$ . 
For comparison, we also show the empirical \ac{pdf} obtained by histogram density estimator and theoretical Gaussian \ac{pdf}. As seen, the derived \ac{pdf} accurately matches the empirical \ac{pdf}. However, it exhibits a larger tail compared to the Gaussian distribution.   
In Table \ref{table: pdfs metrics}, we compare the variance, kurtosis, and skewness of the derived \ac{pdf} in \eqref{eq: pdf}, the empirical histogram density estimation of the \ac{pdf}, and the Gaussian distribution. The kurtosis measures the fourth-order central moment of the random variable $Y_{n_{\rm{I}}}[1]$ with mean $\mu\triangleq\mathbb{E}\{Y_{n_{\rm{I}}}[1]|{\rm{H}}_{0}\}$, and the skewness is a measure of the symmetry in the distribution.  
Large deviation from the mean yields large values of kurtosis. 
For fair comparison, we consider the normalized kurtosis $\kappa$ and skewness $\xi$ defined as 
\begin{align} \label{hjiuoiokl}
\kappa \triangleq \frac{\mathbb{E}\{(Y_{n_{\rm{I}}}[1]-\mu)^4\}}
{\mathbb{E}^2\{(Y_{n_{\rm{I}}}[1]-\mu)^2\}},
\end{align}
and 
\begin{align}
 \xi \triangleq \frac{\mathbb{E}\{(Y_{n_{\rm{I}}}[1]-\mu)^3\}}
{\mathbb{E}^{\frac{3}{2}}\{(Y_{n_{\rm{I}}}[1]-\mu)^2\}}.
\end{align}
To estimate the normalized kurtosis and skewness for the emirical \ac{pdf}, we use  
\begin{align} \label{gfdgdgdkk}
\hat{\kappa} \triangleq \frac{\frac{1}{M}\sum_{n=0}^{M-1}({y_{n_{\rm{I}}}[1]}-\hat{\mu})^4}{\big{(}\frac{1}{M}\sum_{n=0}^{M-1}({y_{n_{\rm{I}}}[1]}-\hat{\mu})^2\big{)}^2},
\end{align}
and
\begin{align} \label{gfdgdgdkk1}
{\hat{\xi}} \triangleq \frac{\frac{1}{M}\sum_{n=0}^{M-1}({y_{n_{\rm{I}}}[1]}-\hat{\mu})^3}{\big{(}\frac{1}{M}\sum_{n=0}^{M-1}({y_{n_{\rm{I}}}[1]}-\hat{\mu})^2\big{)}^{\frac{3}
{2}}},
\end{align}
where 
\begin{align}
\hat{\mu}=\frac{1}{M}\sum_{n=0}^{M-1}{y_{n_{\rm{I}}}[1]}. 
\end{align}
To estimate the kurtosis and skewness in \eqref{gfdgdgdkk} and \eqref{gfdgdgdkk1}, we set $M=10^6$. 
As seen in Table \ref{table: pdfs metrics}, the theoretical kurtosis obtained by \eqref{hjiuoiokl} equals the empirical kurtosis in \eqref{gfdgdgdkk} with precision of $0.01$.
Further, these values are larger than 3; hence,
it indicates a non-Gaussian \ac{pdf}, which in particular, has a larger tail.
In  Table \ref{table: pdfs metrics}, we also observe that for the theoretical and empirical \ac{pdf}s, the  skewness is zero, which implies that the \ac{pdf} is symmetric around its mean.

\begin{figure*}
\begin{align}\label{dpos}
\hspace{-2em}{\bf f}_{\bf{Y}}({\bf{y}} | {\rm{H}}_d)& = 
\prod_{k=0}^{n_{\rm{s}}-d-1}\tilde{f}_{Y[k+d]}
\big{(}y[k]|{\rm{H}}_0\big{)}
\prod_{m=1}^{N-1} \Bigg{(} \prod_{u=0}^{n_{\rm{s}}-1}
\tilde{f}_{Y[u]}
\big{(}y[mn_{\rm{s}}+u-d]|{\rm{H}}_0
\big{)}\Bigg{)}\prod_{v=Nn_{\rm{s}}-d}^{Nn_{\rm{s}}-1}\tilde{f}_{Y[v-Nn_{\rm{s}}+d]}\big{(}y[v]|{\rm{H}}_0\big{)}
,\,\,\,\,\,\,\, \,\,\ d \ge 0 
\end{align}
\begin{align}\label{dpos1}
&{\bf f}_{\bf{Y}}({\bf{y}} | {\rm{H}}_d) =
\prod_{k=0}^{|d|-1}\tilde{f}_{Y[-]}\big{(}y[k]|{\rm{H}}_0\big{)}
\prod_{m=0}^{N-2} \Bigg{(} 
\prod_{u=0}^{n_{\rm{s}}-1}\tilde{f}_{Y[u]}
\big{(}y[mn_{\rm{s}}+u-d]|{\rm{H}}_0\big{)}
\Bigg{)}
\prod_{u=(N-1)n_{\rm{s}}-d}^{Nn_{\rm{s}}-1}
\tilde{f}_{Y[u-(N-1)n_{\rm{s}}+d]}\big{(}y[u]|{\rm{H}}_0\big{)},\,\,\,\,\,\,\,\,\,\,\ d<0
\end{align}
\end{figure*}
\subsection{\ac{pdf} Derivation for Delay Hypothesis ${\rm H}_d$, $d \neq 0$ }
In order to obtain the \ac{ml} estimator, we need to derive the joint \ac{pdf} of the received samples given all delay hypotheses ${\rm{H}}_{d}$, $d\in \{-n_{\rm{s}}+1,\dots,-1,0,1,\dots,$ $n_{\rm{s}}-1\}$.
 In Theorem \ref{theo: pdf y}, we derived the \ac{pdf} of the received samples given hypothesis ${\rm{H}}_{0}$, i.e., $f_{{\bf Y}}({\bf y}|{\rm{H}}_0)$. In Appendix \ref{proof: theo final}, we prove that 
$f_{{\bf Y}}({\bf y}|{\rm{H}}_d)$ can be expressed based on the joint \ac{pdf} of the received samples given ${\rm{H}}_{0}$ as it is shown in \eqref{dpos} and \eqref{dpos1} at the top of this page, where  
\begin{align}\label{eq:29}
 \tilde{f}_{Y[m]}( y| {\rm{H}}_0 ) &\triangleq 
 {f}_{Y[nn_{\rm{s}}+m]}( y| {\rm{H}}_0 )=f_{Y_n[m]}( y| {\rm{H}}_0)
 \\ \nonumber 
& \approx f_{Y_{n_{\rm{I}}}[m]} \big( y_{\rm{I}} | {\rm{H}}_0 \big) f_{Y_{n_{\rm{Q}} }[m]} \big( y_{\rm{Q}} | {\rm{H}}_0 \big) \\ \nonumber 
&  \triangleq f_{Y_{{\rm{I}}}[m]} \big( y_{\rm{I}} | {\rm{H}}_0 \big)f_{Y_{{\rm{Q}}}[m]} \big( y_{\rm{Q}} | {\rm{H}}_0 \big),
\,\,\,\,\,\,\,\,\,\,\,\ n \ge 0
\end{align}
for $0 \le m \le n_{\rm{s}}-1$, and
\begin{align}\label{eq:30}
\tilde{f}_{Y[-]} (y|{\rm{H}}_0) & \triangleq {f}_{Y[nn_{\rm{s}}+m]}( y| {\rm{H}}_0 )=f_{Y_n[m]}( y| {\rm{H}}_0),
\\ \nonumber  
& \approx f_{Y_{n_{\rm{I}}}[m]}(y_{\rm{I}}|{\rm{H}}_0) f_{Y_{n_{\rm{Q}}}[m]}(y_{\rm{Q}}|{\rm{H}}_0) \\ \nonumber 
& \triangleq f_{Y_{{\rm{I}}}[-]} \big( y_{\rm{I}} | {\rm{H}}_0 \big)f_{Y_{{\rm{Q}}}[-]} \big( y_{\rm{Q}} | {\rm{H}}_0 \big),
\,\,\,\,\,\,\,\,\,\,\ n<0,
\end{align}
where $y \triangleq y_{\rm{I}}+iy_{\rm{Q}}$, $f_{Y[nn_{\rm{s}}+m]}(\cdot | {\rm{H}_0})$ is the \ac{pdf} of the received sample $y[nn_{\rm{s}}+m]  \triangleq y_{n}[m]$ given ${\rm{H}_0}$, and the \ac{pdf} of $\Re\{y[nn_{\rm{s}}+m]\}={y_{n_{\rm{I}}}}[m]$, i.e., $f_{Y_{n_{\rm{I}}}[m]}(\cdot|{\rm{H}}_0)$ and the \ac{pdf} of $\Im\{y[nn_{\rm{s}}+m]\}={y_{n_{\rm{Q}}}}[m]$, i.e., $f_{Y_{n_{\rm{Q}}}[m]}(\cdot|{\rm{H}}_0)$ are given in Theorem \ref{theo: pdf y}.
The relation between the \ac{pdf} of the received samples given ${\rm{H}}_d$ and ${\rm{H}}_0$ is attributed to the 
periodicity of the zero-padded guard interval.

To visualize \eqref{dpos} and \eqref{dpos1} of the revised manuscript,
let us consider the vector of PDF in  \eqref{eq: pdf matrix H0} in the revised manuscript.
We consider the first shown dashed line from the top in the PDF vector as reference line. The elements below this reference line in PDF vector are periodic with period $n_{\rm{s}}$  (see the pattern in Fig. \ref{fig: conv}  in the revised manuscript). The elements above this reference line represent the PDF of the noise samples. For $d\ge0$, the elements of the observation vector $\bf{y}$ are respectively
substituted in the PDF vector starting from the $(d+1)$-th element below the reference  line. This results in  \eqref{dpos}. Similarly, for $d<0$, the elements of the observation vector $\bf{y}$ are respectively substituted in the PDF vector starting $|d|$ elements above the reference  line. This results in \eqref{dpos1}.

\subsection{\ac{ml} TO Estimator}
The ML estimation for TO is defined to be the value of $d$ that maximizes $f_{{\bf Y}}({\bf y}|{\rm{H}}_d)$ for ${\bf y}$ fixed, i.e., the value that maximizes the likelihood function. The maximization is performed over the allowable range of $d$. Corollary \ref{theo: final} summarizes the proposed  NDA-ML TO estimation for ZP-OFDM. 

\begin{corollary} \label{theo: final}
For a \ac{zp}-OFDM system in a doubly-selective fading channel with the received vector $\bf{y}$ in \eqref{eq:14}, the NDA-\ac{ml} \ac{to} estimator is given by
\begin{equation}\label{8989ioio}
\hat{d}^{\text{opt}}=  \operatorname*{argmax}_{d \in \{-n_{\rm{s}}+1,..., n_{\rm{s}}-1\}} {\bf f}_{\bf{Y}}({\bf{y}} | {\rm{H}}_d), 
\end{equation}
\noindent where ${\bf f}_{\bf{Y}}({\bf{y}} | {\rm{H}}_d)$ is given in \eqref{dpos} and \eqref{dpos1}. 
\end{corollary}

The proposed time synchronization method can be extended to ZP-OFDM with non-rectangular pulse shaping, but it requires the modification of the PDF in equation (20). In this case,
the $T_{\rm sa}$-spaced  baseband  OFDM  samples are modeled as independent random variables with different variances, which makes the derivation of the PDF challenging.

Since a closed-form expression cannot be found for the ML estimator in \eqref{8989ioio}, a numerical approach can be used. Numerical methods employ either an exhaustive 
search or an iterative maximization of the likelihood function.

\begin{table}[t!]
\vspace{-1em}
\centering 
 \caption{Statistical Analysis}
\label{table: pdfs metrics}
\resizebox{0.45\textwidth}{!}{
\begin{tabularx}{0.36\textwidth}{lccc}
\toprule
Metric & Empirical    & Analytical & Gaussian \\ 
\midrule
Mean   & 4.6250 $\times 10^{-4}$      &     0      & 0  \\ 
Variance & 0.3206        & 0.3205       & 0.3206  \\ 
Skewness       & 0.0031       & 0          & 0  \\ 
Kurtosis        & 4.5315      &  4.5653         & 3   \\ 
\bottomrule
\end{tabularx} }
\vspace{-1em}
\end{table}

\section{Low-Complexity Implementation}
\label{sec: importance samp} 

The derived \ac{pdf} in \eqref{eq: pdf} is complex due to the integral terms including the Gaussian error function $\Phi(\cdot)$. Hence, practical implementation of the proposed \ac{ml} can be challenging. 
An alternative approach with feasible implementation
and lower complexity is to employ MCS techniques to approximate the joint \ac{pdf} of the received samples.
MCS methods benefit from the availability of computer generated random variables to approximate univariate and multidimensional integrals in Bayesian estimation, inference, and optimization problems.    
The key idea behind MCS method is to generate independent random samples from a \ac{pdf} usually known up to a normalizing constant.  In the following discussion, we use MCS integration method in order to
make efficient implementation of the proposed theoretical  
 NDA-\ac{ml} estimator in \eqref{8989ioio} possible.

\subsection{MCS Method}
In Appendix \ref{proof: theo}, we proved that ${f}_{Y_{{\rm{I}}}[m]}(y_{\rm{I}} | {\rm{H}}_0)$ in \eqref{eq: pdf} is expressed in an integral form as follows
\begin{equation} \label{eq: importance sampling}
\begin{split}
&f_{Y_{n_{\rm{I}}}[m]}(y_{\rm{I}} | {\rm{H}}_0) = \int_{-\infty}^{\infty} f_{W_{n_{\rm{I}}}[m]}(y_{\rm{I}}-v) f_{V_{n_{\rm{I}}}[m]}(v | {\rm{H}}_0) dv \\ 
&= \int^{\infty}_{-\infty}  \frac{1}{\sqrt{\pi \sigma^2_{\rm{w}}}} \exp\bigg\{-\frac{1}{\sigma^2_{\rm{w}}}   \big(y_{\rm{I}} - v\big)^2 \bigg\} f_{V_{n_{\rm{I}}}[m]} (v) dv,
\end{split}
\end{equation}
where
\begin{align} \nonumber
f_{V_{n_{\rm{I}}}[m]}(v | \rm{H}_0)& =    \Bigg( \prod_{i=a}^{b} \lambda_i \Bigg)^2 \sum_{j=a}^{b}  \sum_{n=a}^{b} \frac{ 1}{ \prod_{k=a, k \neq j}^{b} (\lambda_k - \lambda_j)} \\ \label{pdf_v0}
&~ \times\frac{ 1 }{ \prod_{p=a, p \neq j}^{b} (\lambda_p - \lambda_n)} \frac{e^{-\lambda_j |v|}}{\lambda_j + \lambda_n},
\end{align}
and $f_{W_{n_{\rm{I}}}[m]} (w)$ is the \ac{pdf} of the white Gaussian noise with variance $\sigma_{\rm{w}}^2/2$.  Generating samples from random variable $V_{n_{\rm{I}}}[m]$ with \ac{pdf} in
\eqref{pdf_v0} is straightforward since it is expressed as a linear function of independent exponentially distributed random variables with  rate parameter $\lambda_k=(\sigma_{{\rm{h}}_k} \sigma_{\rm{x}}/2)^{-1}$ as shown in \eqref{eq: decomposee} of Appendix \ref{proof: theo}.

By using Monte Carlo integration method, we can write 
\begin{align}\label{eq: monte c}
\hspace{-1em}
f_{Y_{n_{\rm{I}}}[m]}(y_{\rm{I}} | {\rm{H}}_0) \simeq \frac{1}{L}\sum^{L-1}_{\ell=0}   \frac{1}{\sqrt{\pi \sigma^2_{\rm{w}}}} \exp\bigg\{-\frac{1}{\sigma^2_{\rm{w}}}  \big(y_{\rm{I}} -  v_\ell\big)^2 \bigg\}, \end{align}
 where $L$ is the number of  Monte  Carlo samples, and  $\{v_0,v_1, \dots,v_{L-1}\}$ are \ac{iid} samples drawn from $V_{n_{\rm{I}}}[m]$. 
 By applying Monte Carlo integration to the marginal \ac{pdf}s of the in-phase and quadrature components of  
 ${\bf{y}}$, the joint \ac{pdf} is given by 
\begin{align}\nonumber 
&    f_{\bf{Y}}( {\bf y}| {\rm{H}}_{0})  
\hspace{-0.11em} \approx \hspace{-0.21em}\prod_{n=0}^{N-1}
\prod_{m=0}^{n_{{\rm{s}}}-1} \hspace{-0.11em} f_{Y_{n_{\rm{I}}}[m]}(y_{n_{{\rm{I}}}}[m] |  {\rm{H}}_{0}) f_{Y_{n_{\rm{Q}}}[m]}(y_{n_{{\rm{Q}}}}[m] |  {\rm{H}}_{0}) \\ \nonumber 
& \simeq \frac{1}{({\pi \sigma^2_{\rm{w}}})^{Nn_{\rm{s}}}}  \prod_{n=0}^{N-1}
\prod_{m=0}^{n_{{\rm{s}}}-1} 
 \Bigg{(} \sum^{L}_{\ell=1}    \exp\Bigg\{\frac{-1}{\sigma^2_{\rm{w}}}  \big(y_{n_{{\rm{I}}}}[m] -  v_{{n_{{\rm{I}}}}}^m[\ell]\big)^2 \bigg\} \\ \label{eq: monte c}
& \hspace{0.5cm}\times 
 \sum^{L}_{\ell=1}    \exp\Bigg\{\frac{-1}{\sigma^2_{\rm{w}}}  \big(y_{n_{{\rm{Q}}}}[m] -  v_{{n_{{\rm{Q}}}}}^m[\ell]\big)^2 \bigg\}\Bigg{)}, 
\end{align}
where $v_{{n_{{\rm{I}}}}}^m[\ell]$ and $v_{{n_{{\rm{Q}}}}}^m[\ell]$
are \ac{iid} values drawn from random variables with \ac{pdf}s $f_{V_{n_{\rm{I}}}[m]} (v)$ and $f_{V_{n_{\rm{Q}}}[m]} (v)$, respectively. 
By using \eqref{eq: monte c}, we can design the MCS implementation of the theoretical \ac{nda}-ML \ac{to} estimator in \eqref{8989ioio} as in   
Fig. \ref{fig: gamma block}.

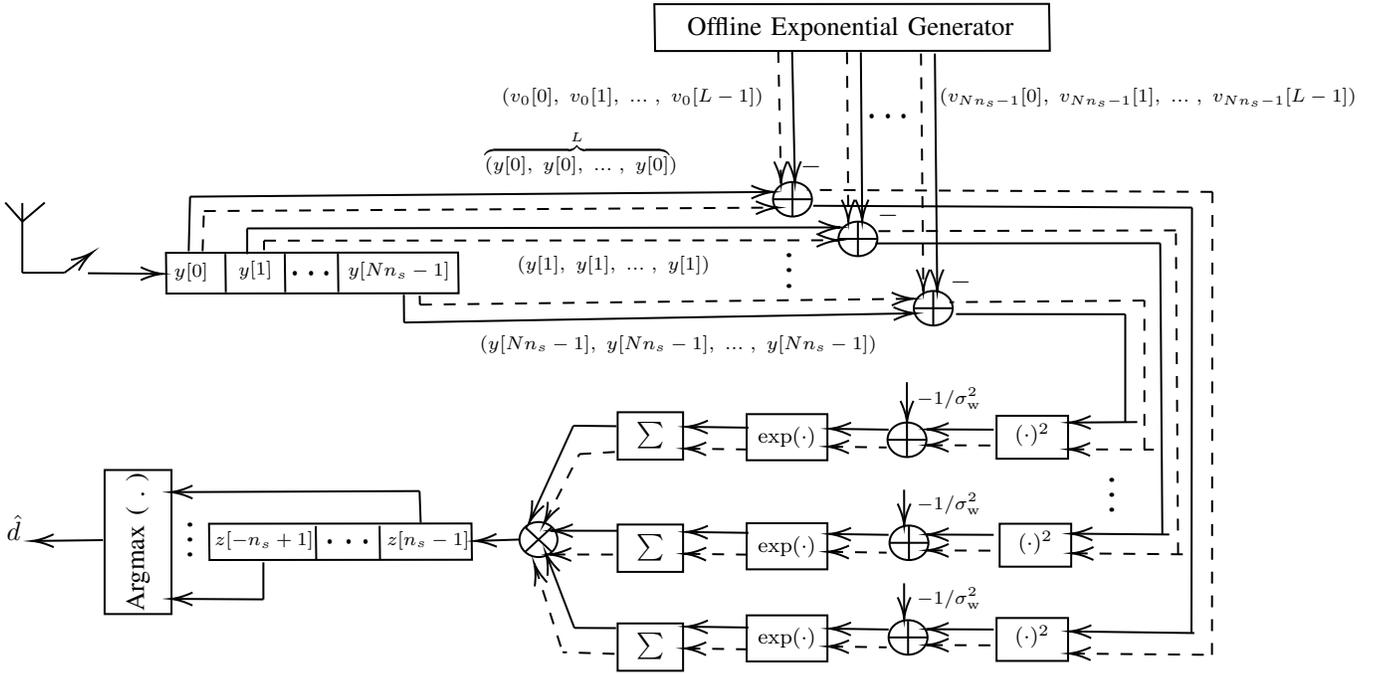
\begin{figure*}[t]
\centering

\tikzset{every picture/.style={line width=0.75pt}} 

\begin{tikzpicture}[x=0.75pt,y=0.75pt,yscale=-1,xscale=1]

\draw  [line width=0.75]  (106.48,333.63) -- (106.25,315.32) -- (238.62,315.18) -- (238.86,333.49) -- cycle ;
\draw [line width=0.75]    (160.33,333.74) -- (160.1,316.48) ;
\draw [line width=0.75]    (192.49,333.64) -- (192.41,327.22) -- (192.27,316.38) ;
\draw  [line width=0.75]  (377.41,260) -- (418.12,260) -- (418.12,283.27) -- (377.41,283.27) -- cycle ;
\draw  [line width=0.75]  (312.33,259.04) -- (345.34,259.04) -- (345.34,282.94) -- (312.33,282.94) -- cycle ;
\draw  [line width=0.75]  (312.33,315.72) -- (345.34,315.72) -- (345.34,339.62) -- (312.33,339.62) -- cycle ;
\draw  [line width=0.75]  (312.33,365.63) -- (345.34,365.63) -- (345.34,389.53) -- (312.33,389.53) -- cycle ;
\draw [line width=0.75]    (348.86,266.41) -- (378.33,266.83) ;
\draw [shift={(346.86,266.38)}, rotate = 0.81] [color={rgb, 255:red, 0; green, 0; blue, 0 }  ][line width=0.75]    (10.93,-3.29) .. controls (6.95,-1.4) and (3.31,-0.3) .. (0,0) .. controls (3.31,0.3) and (6.95,1.4) .. (10.93,3.29)   ;
\draw  [line width=0.75]  (377.41,315.06) -- (418.12,315.06) -- (418.12,338.33) -- (377.41,338.33) -- cycle ;
\draw [line width=0.75]    (348.86,319.5) -- (378.33,319.92) ;
\draw [shift={(346.86,319.47)}, rotate = 0.81] [color={rgb, 255:red, 0; green, 0; blue, 0 }  ][line width=0.75]    (10.93,-3.29) .. controls (6.95,-1.4) and (3.31,-0.3) .. (0,0) .. controls (3.31,0.3) and (6.95,1.4) .. (10.93,3.29)   ;
\draw  [line width=0.75]  (377.41,361.56) -- (418.12,361.56) -- (418.12,384.83) -- (377.41,384.83) -- cycle ;
\draw [line width=0.75]    (348.86,367.28) -- (378.33,367.7) ;
\draw [shift={(346.86,367.25)}, rotate = 0.81] [color={rgb, 255:red, 0; green, 0; blue, 0 }  ][line width=0.75]    (10.93,-3.29) .. controls (6.95,-1.4) and (3.31,-0.3) .. (0,0) .. controls (3.31,0.3) and (6.95,1.4) .. (10.93,3.29)   ;
\draw  [line width=0.75]  (262.84,323.29) .. controls (262.84,318.41) and (267.13,314.45) .. (272.42,314.45) .. controls (277.71,314.45) and (282,318.41) .. (282,323.29) .. controls (282,328.18) and (277.71,332.14) .. (272.42,332.14) .. controls (267.13,332.14) and (262.84,328.18) .. (262.84,323.29) -- cycle ; \draw  [line width=0.75]  (265.65,317.04) -- (279.2,329.55) ; \draw  [line width=0.75]  (279.2,317.04) -- (265.65,329.55) ;
\draw [line width=0.75]    (268.5,313.01) -- (290.02,265.85) ;
\draw [shift={(267.67,314.83)}, rotate = 294.53] [color={rgb, 255:red, 0; green, 0; blue, 0 }  ][line width=0.75]    (10.93,-3.29) .. controls (6.95,-1.4) and (3.31,-0.3) .. (0,0) .. controls (3.31,0.3) and (6.95,1.4) .. (10.93,3.29)   ;
\draw [line width=0.75]    (311.71,266.09) -- (290.02,265.85) ;
\draw [line width=0.75]    (311.54,368) -- (290.51,367.9) ;
\draw [line width=0.75]    (276.84,332.04) -- (290.51,367.9) ;
\draw [shift={(276.12,330.17)}, rotate = 69.14] [color={rgb, 255:red, 0; green, 0; blue, 0 }  ][line width=0.75]    (10.93,-3.29) .. controls (6.95,-1.4) and (3.31,-0.3) .. (0,0) .. controls (3.31,0.3) and (6.95,1.4) .. (10.93,3.29)   ;
\draw [line width=0.75]    (280.77,318.84) -- (312.33,318.88) ;
\draw [shift={(278.77,318.83)}, rotate = 0.09] [color={rgb, 255:red, 0; green, 0; blue, 0 }  ][line width=0.75]    (10.93,-3.29) .. controls (6.95,-1.4) and (3.31,-0.3) .. (0,0) .. controls (3.31,0.3) and (6.95,1.4) .. (10.93,3.29)   ;
\draw [line width=0.75]    (458.26,262.24) -- (458.23,243.89) ;
\draw [shift={(458.27,264.24)}, rotate = 269.89] [color={rgb, 255:red, 0; green, 0; blue, 0 }  ][line width=0.75]    (10.93,-3.29) .. controls (6.95,-1.4) and (3.31,-0.3) .. (0,0) .. controls (3.31,0.3) and (6.95,1.4) .. (10.93,3.29)   ;
\draw [line width=0.75]    (421.06,267.41) -- (449,267.83) ;
\draw [shift={(419.06,267.38)}, rotate = 0.85] [color={rgb, 255:red, 0; green, 0; blue, 0 }  ][line width=0.75]    (10.93,-3.29) .. controls (6.95,-1.4) and (3.31,-0.3) .. (0,0) .. controls (3.31,0.3) and (6.95,1.4) .. (10.93,3.29)   ;
\draw [line width=0.75]    (421.06,320.5) -- (449,320.92) ;
\draw [shift={(419.06,320.47)}, rotate = 0.85] [color={rgb, 255:red, 0; green, 0; blue, 0 }  ][line width=0.75]    (10.93,-3.29) .. controls (6.95,-1.4) and (3.31,-0.3) .. (0,0) .. controls (3.31,0.3) and (6.95,1.4) .. (10.93,3.29)   ;
\draw [line width=0.75]    (421.06,368.28) -- (449,368.7) ;
\draw [shift={(419.06,368.25)}, rotate = 0.85] [color={rgb, 255:red, 0; green, 0; blue, 0 }  ][line width=0.75]    (10.93,-3.29) .. controls (6.95,-1.4) and (3.31,-0.3) .. (0,0) .. controls (3.31,0.3) and (6.95,1.4) .. (10.93,3.29)   ;

\draw [line width=0.75]    (456.84,314.15) -- (456.81,298.17) ;
\draw [shift={(456.84,316.15)}, rotate = 269.89] [color={rgb, 255:red, 0; green, 0; blue, 0 }  ][line width=0.75]    (10.93,-3.29) .. controls (6.95,-1.4) and (3.31,-0.3) .. (0,0) .. controls (3.31,0.3) and (6.95,1.4) .. (10.93,3.29)   ;
\draw [line width=0.75]    (456.84,361.36) -- (456.81,345.38) ;
\draw [shift={(456.84,363.36)}, rotate = 269.89] [color={rgb, 255:red, 0; green, 0; blue, 0 }  ][line width=0.75]    (10.93,-3.29) .. controls (6.95,-1.4) and (3.31,-0.3) .. (0,0) .. controls (3.31,0.3) and (6.95,1.4) .. (10.93,3.29)   ;
\draw [line width=0.75]    (242.37,323.97) -- (262.84,323.29) ;
\draw [shift={(240.37,324.03)}, rotate = 358.12] [color={rgb, 255:red, 0; green, 0; blue, 0 }  ][line width=0.75]    (10.93,-3.29) .. controls (6.95,-1.4) and (3.31,-0.3) .. (0,0) .. controls (3.31,0.3) and (6.95,1.4) .. (10.93,3.29)   ;
\draw  [line width=0.75]  (53.6,288.2) -- (87.31,288.2) -- (87.31,360.9) -- (53.6,360.9) -- cycle ;

\draw  [line width=0.75]  (84.79,199.13) -- (84.54,178.58) -- (231.88,178.43) -- (232.13,198.98) -- cycle ;
\draw [line width=0.75]    (114.72,198.51) -- (114.49,179.14) ;
\draw [line width=0.75]    (144.7,198.48) -- (144.46,179.11) ;
\draw [line width=0.75]    (171.49,198.41) -- (171.4,191.21) -- (171.25,179.04) ;
\draw [line width=0.75]    (12.1,153.49) -- (12.1,188.64) ;
\draw [line width=0.75]    (45.15,188.89) -- (80.73,189.17) ;
\draw [shift={(82.73,189.19)}, rotate = 180.45] [color={rgb, 255:red, 0; green, 0; blue, 0 }  ][line width=0.75]    (10.93,-3.29) .. controls (6.95,-1.4) and (3.31,-0.3) .. (0,0) .. controls (3.31,0.3) and (6.95,1.4) .. (10.93,3.29)   ;
\draw [line width=0.75]    (12.1,162.85) -- (23.42,153.2) ;
\draw [line width=0.75]    (12.1,162.85) -- (3.5,153.2) ;
\draw [line width=0.75]    (12.1,188.64) -- (33.38,188.6) ;
\draw [line width=0.75]    (33.38,188.6) -- (45.8,179.54) ;
\draw [shift={(47.42,178.36)}, rotate = 503.89] [color={rgb, 255:red, 0; green, 0; blue, 0 }  ][line width=0.75]    (10.93,-3.29) .. controls (6.95,-1.4) and (3.31,-0.3) .. (0,0) .. controls (3.31,0.3) and (6.95,1.4) .. (10.93,3.29)   ;
\draw [line width=0.75]    (204.63,213.75) -- (204.5,199.72) ;
\draw [line width=0.75]    (204.63,213.75) -- (460.46,209.13) ;
\draw [shift={(462.46,209.09)}, rotate = 538.97] [color={rgb, 255:red, 0; green, 0; blue, 0 }  ][line width=0.75]    (10.93,-3.29) .. controls (6.95,-1.4) and (3.31,-0.3) .. (0,0) .. controls (3.31,0.3) and (6.95,1.4) .. (10.93,3.29)   ;

\draw [line width=0.75]    (96,177.75) -- (96.69,150.24) ;
\draw [line width=0.75]    (96.69,150.24) -- (388.7,147.86) ;
\draw [shift={(390.7,147.84)}, rotate = 539.53] [color={rgb, 255:red, 0; green, 0; blue, 0 }  ][line width=0.75]    (10.93,-3.29) .. controls (6.95,-1.4) and (3.31,-0.3) .. (0,0) .. controls (3.31,0.3) and (6.95,1.4) .. (10.93,3.29)   ;

\draw [line width=0.75]    (125.5,179.24) -- (125.5,166.26) ;
\draw [line width=0.75]    (125.5,166.26) -- (422.63,165.75) ;
\draw [shift={(424.63,165.75)}, rotate = 539.9] [color={rgb, 255:red, 0; green, 0; blue, 0 }  ][line width=0.75]    (10.93,-3.29) .. controls (6.95,-1.4) and (3.31,-0.3) .. (0,0) .. controls (3.31,0.3) and (6.95,1.4) .. (10.93,3.29)   ;

\draw  [line width=0.75]  (331.24,76.73) -- (330.85,53.07) -- (529.92,52.93) -- (530.3,76.59) -- cycle ;
\draw [line width=0.75]    (400.33,76.67) -- (401.71,141.44) ;
\draw [shift={(401.76,143.44)}, rotate = 268.78] [color={rgb, 255:red, 0; green, 0; blue, 0 }  ][line width=0.75]    (10.93,-3.29) .. controls (6.95,-1.4) and (3.31,-0.3) .. (0,0) .. controls (3.31,0.3) and (6.95,1.4) .. (10.93,3.29)   ;
\draw [line width=0.75]    (435,77.33) -- (435.67,160.71) ;
\draw [shift={(435.69,162.71)}, rotate = 269.54] [color={rgb, 255:red, 0; green, 0; blue, 0 }  ][line width=0.75]    (10.93,-3.29) .. controls (6.95,-1.4) and (3.31,-0.3) .. (0,0) .. controls (3.31,0.3) and (6.95,1.4) .. (10.93,3.29)   ;
\draw [line width=0.75]    (472.33,78) -- (473.5,195.97) ;
\draw [shift={(473.52,197.97)}, rotate = 269.43] [color={rgb, 255:red, 0; green, 0; blue, 0 }  ][line width=0.75]    (10.93,-3.29) .. controls (6.95,-1.4) and (3.31,-0.3) .. (0,0) .. controls (3.31,0.3) and (6.95,1.4) .. (10.93,3.29)   ;
\draw [line width=0.75]    (601.75,370.5) -- (602,155.74) ;
\draw [line width=0.75]    (602,155.74) -- (411,154.83) ;
\draw [line width=0.75]    (585.98,173.75) -- (442.33,173.83) ;
\draw [line width=0.75]    (587.05,320.75) -- (585.98,173.75) ;
\draw [line width=0.75]    (568.35,209.58) -- (483,209.5) ;
\draw [line width=0.75]    (568.35,264) -- (568.35,209.58) ;
\draw [line width=0.75]    (541.96,264.24) -- (574.42,264) ;
\draw [shift={(539.96,264.25)}, rotate = 359.59] [color={rgb, 255:red, 0; green, 0; blue, 0 }  ][line width=0.75]    (10.93,-3.29) .. controls (6.95,-1.4) and (3.31,-0.3) .. (0,0) .. controls (3.31,0.3) and (6.95,1.4) .. (10.93,3.29)   ;
\draw [line width=0.75]    (542.82,320.35) -- (590.53,320.75) ;
\draw [shift={(540.82,320.34)}, rotate = 0.48] [color={rgb, 255:red, 0; green, 0; blue, 0 }  ][line width=0.75]    (10.93,-3.29) .. controls (6.95,-1.4) and (3.31,-0.3) .. (0,0) .. controls (3.31,0.3) and (6.95,1.4) .. (10.93,3.29)   ;
\draw [line width=0.75]    (542.28,369.68) -- (603.19,370.5) ;
\draw [shift={(540.28,369.65)}, rotate = 0.78] [color={rgb, 255:red, 0; green, 0; blue, 0 }  ][line width=0.75]    (10.93,-3.29) .. controls (6.95,-1.4) and (3.31,-0.3) .. (0,0) .. controls (3.31,0.3) and (6.95,1.4) .. (10.93,3.29)   ;

\draw  [line width=0.75]  (503.42,260.83) -- (539.51,260.83) -- (539.51,282.5) -- (503.42,282.5) -- cycle ;
\draw [line width=0.75]    (470.33,267.69) -- (502,267.83) ;
\draw [shift={(468.33,267.68)}, rotate = 0.25] [color={rgb, 255:red, 0; green, 0; blue, 0 }  ][line width=0.75]    (10.93,-3.29) .. controls (6.95,-1.4) and (3.31,-0.3) .. (0,0) .. controls (3.31,0.3) and (6.95,1.4) .. (10.93,3.29)   ;
\draw [line width=0.75]    (470.33,320.78) -- (502,320.92) ;
\draw [shift={(468.33,320.77)}, rotate = 0.25] [color={rgb, 255:red, 0; green, 0; blue, 0 }  ][line width=0.75]    (10.93,-3.29) .. controls (6.95,-1.4) and (3.31,-0.3) .. (0,0) .. controls (3.31,0.3) and (6.95,1.4) .. (10.93,3.29)   ;
\draw [line width=0.75]    (470.33,368.56) -- (502,368.7) ;
\draw [shift={(468.33,368.55)}, rotate = 0.25] [color={rgb, 255:red, 0; green, 0; blue, 0 }  ][line width=0.75]    (10.93,-3.29) .. controls (6.95,-1.4) and (3.31,-0.3) .. (0,0) .. controls (3.31,0.3) and (6.95,1.4) .. (10.93,3.29)   ;

\draw  [line width=0.75]  (504.99,315.32) -- (541.08,315.32) -- (541.08,336.99) -- (504.99,336.99) -- cycle ;
\draw  [line width=0.75]  (503.42,362.89) -- (539.51,362.89) -- (539.51,384.56) -- (503.42,384.56) -- cycle ;
\draw [line width=0.75]    (89,353.25) -- (133.67,353.5) ;
\draw [shift={(87,353.24)}, rotate = 0.32] [color={rgb, 255:red, 0; green, 0; blue, 0 }  ][line width=0.75]    (10.93,-3.29) .. controls (6.95,-1.4) and (3.31,-0.3) .. (0,0) .. controls (3.31,0.3) and (6.95,1.4) .. (10.93,3.29)   ;
\draw [line width=0.75]    (89,299.36) -- (212.42,298.83) ;
\draw [shift={(87,299.37)}, rotate = 359.76] [color={rgb, 255:red, 0; green, 0; blue, 0 }  ][line width=0.75]    (10.93,-3.29) .. controls (6.95,-1.4) and (3.31,-0.3) .. (0,0) .. controls (3.31,0.3) and (6.95,1.4) .. (10.93,3.29)   ;
\draw [line width=0.75]    (213,314.74) -- (212.42,298.83) ;
\draw [line width=0.75]    (133.67,353.5) -- (133.67,334.8) ;

\draw [line width=0.75]    (18.97,324.01) -- (52.76,323.58) ;
\draw [shift={(16.97,324.03)}, rotate = 359.27] [color={rgb, 255:red, 0; green, 0; blue, 0 }  ][line width=0.75]    (10.93,-3.29) .. controls (6.95,-1.4) and (3.31,-0.3) .. (0,0) .. controls (3.31,0.3) and (6.95,1.4) .. (10.93,3.29)   ;

\draw  [line width=0.75]  (390.7,151.54) .. controls (390.7,146.73) and (395.09,142.83) .. (400.51,142.83) .. controls (405.94,142.83) and (410.33,146.73) .. (410.33,151.54) .. controls (410.33,156.35) and (405.94,160.25) .. (400.51,160.25) .. controls (395.09,160.25) and (390.7,156.35) .. (390.7,151.54) -- cycle ; \draw  [line width=0.75]  (390.7,151.54) -- (410.33,151.54) ; \draw  [line width=0.75]  (400.51,142.83) -- (400.51,160.25) ;
\draw [line width=0.75]  [dash pattern={on 4.5pt off 4.5pt}]  (393.33,76.67) -- (394.71,141.44) ;
\draw [shift={(394.76,143.44)}, rotate = 268.78] [color={rgb, 255:red, 0; green, 0; blue, 0 }  ][line width=0.75]    (10.93,-3.29) .. controls (6.95,-1.4) and (3.31,-0.3) .. (0,0) .. controls (3.31,0.3) and (6.95,1.4) .. (10.93,3.29)   ;
\draw [line width=0.75]  [dash pattern={on 4.5pt off 4.5pt}]  (428,77.33) -- (428.67,160.71) ;
\draw [shift={(428.69,162.71)}, rotate = 269.54] [color={rgb, 255:red, 0; green, 0; blue, 0 }  ][line width=0.75]    (10.93,-3.29) .. controls (6.95,-1.4) and (3.31,-0.3) .. (0,0) .. controls (3.31,0.3) and (6.95,1.4) .. (10.93,3.29)   ;
\draw [line width=0.75]  [dash pattern={on 4.5pt off 4.5pt}]  (465.33,78) -- (466.5,195.97) ;
\draw [shift={(466.52,197.97)}, rotate = 269.43] [color={rgb, 255:red, 0; green, 0; blue, 0 }  ][line width=0.75]    (10.93,-3.29) .. controls (6.95,-1.4) and (3.31,-0.3) .. (0,0) .. controls (3.31,0.3) and (6.95,1.4) .. (10.93,3.29)   ;
\draw [line width=0.75]  [dash pattern={on 4.5pt off 4.5pt}]  (103,177.75) -- (103.68,157.51) ;
\draw [line width=0.75]  [dash pattern={on 4.5pt off 4.5pt}]  (103.68,157.51) -- (391,155.76) ;
\draw [shift={(393,155.75)}, rotate = 539.65] [color={rgb, 255:red, 0; green, 0; blue, 0 }  ][line width=0.75]    (10.93,-3.29) .. controls (6.95,-1.4) and (3.31,-0.3) .. (0,0) .. controls (3.31,0.3) and (6.95,1.4) .. (10.93,3.29)   ;
\draw [line width=0.75]  [dash pattern={on 4.5pt off 4.5pt}]  (134,177.25) -- (134,172.44) ;
\draw [line width=0.75]  [dash pattern={on 4.5pt off 4.5pt}]  (134,172.44) -- (421.63,172.25) ;
\draw [shift={(423.63,172.25)}, rotate = 539.96] [color={rgb, 255:red, 0; green, 0; blue, 0 }  ][line width=0.75]    (10.93,-3.29) .. controls (6.95,-1.4) and (3.31,-0.3) .. (0,0) .. controls (3.31,0.3) and (6.95,1.4) .. (10.93,3.29)   ;
\draw [line width=0.75]  [dash pattern={on 4.5pt off 4.5pt}]  (212.63,204.75) -- (212.5,200.22) ;
\draw [line width=0.75]  [dash pattern={on 4.5pt off 4.5pt}]  (212.63,204.75) -- (459.5,201.77) ;
\draw [shift={(461.5,201.75)}, rotate = 539.31] [color={rgb, 255:red, 0; green, 0; blue, 0 }  ][line width=0.75]    (10.93,-3.29) .. controls (6.95,-1.4) and (3.31,-0.3) .. (0,0) .. controls (3.31,0.3) and (6.95,1.4) .. (10.93,3.29)   ;
\draw [line width=0.75]  [dash pattern={on 4.5pt off 4.5pt}]  (578,202.65) -- (479.67,203.5) ;
\draw [line width=0.75]  [dash pattern={on 4.5pt off 4.5pt}]  (578,277.91) -- (578,202.65) ;
\draw [line width=0.75]  [dash pattern={on 4.5pt off 4.5pt}]  (541.67,277.76) -- (582.73,277.91) ;
\draw [shift={(539.67,277.75)}, rotate = 0.21] [color={rgb, 255:red, 0; green, 0; blue, 0 }  ][line width=0.75]    (10.93,-3.29) .. controls (6.95,-1.4) and (3.31,-0.3) .. (0,0) .. controls (3.31,0.3) and (6.95,1.4) .. (10.93,3.29)   ;
\draw [line width=0.75]  [dash pattern={on 4.5pt off 4.5pt}]  (593.88,167.25) -- (441.25,167.65) ;
\draw [line width=0.75]  [dash pattern={on 4.5pt off 4.5pt}]  (595,330.75) -- (593.88,167.25) ;
\draw [line width=0.75]  [dash pattern={on 4.5pt off 4.5pt}]  (543.39,330.27) -- (597.37,330.75) ;
\draw [shift={(541.39,330.25)}, rotate = 0.51] [color={rgb, 255:red, 0; green, 0; blue, 0 }  ][line width=0.75]    (10.93,-3.29) .. controls (6.95,-1.4) and (3.31,-0.3) .. (0,0) .. controls (3.31,0.3) and (6.95,1.4) .. (10.93,3.29)   ;
\draw [line width=0.75]  [dash pattern={on 4.5pt off 4.5pt}]  (612.08,381.5) -- (612.33,147.82) ;
\draw [line width=0.75]  [dash pattern={on 4.5pt off 4.5pt}]  (612.33,147.82) -- (407.82,146.83) ;
\draw [line width=0.75]  [dash pattern={on 4.5pt off 4.5pt}]  (541.95,380.85) -- (612.08,381.5) ;
\draw [shift={(539.95,380.83)}, rotate = 0.53] [color={rgb, 255:red, 0; green, 0; blue, 0 }  ][line width=0.75]    (10.93,-3.29) .. controls (6.95,-1.4) and (3.31,-0.3) .. (0,0) .. controls (3.31,0.3) and (6.95,1.4) .. (10.93,3.29)   ;
\draw [line width=0.75]  [dash pattern={on 4.5pt off 4.5pt}]  (470.33,275.69) -- (504.33,275.83) ;
\draw [shift={(468.33,275.68)}, rotate = 0.23] [color={rgb, 255:red, 0; green, 0; blue, 0 }  ][line width=0.75]    (10.93,-3.29) .. controls (6.95,-1.4) and (3.31,-0.3) .. (0,0) .. controls (3.31,0.3) and (6.95,1.4) .. (10.93,3.29)   ;
\draw [line width=0.75]  [dash pattern={on 4.5pt off 4.5pt}]  (470.33,328.78) -- (504.33,328.92) ;
\draw [shift={(468.33,328.77)}, rotate = 0.23] [color={rgb, 255:red, 0; green, 0; blue, 0 }  ][line width=0.75]    (10.93,-3.29) .. controls (6.95,-1.4) and (3.31,-0.3) .. (0,0) .. controls (3.31,0.3) and (6.95,1.4) .. (10.93,3.29)   ;
\draw [line width=0.75]  [dash pattern={on 4.5pt off 4.5pt}]  (470.33,376.56) -- (504.33,376.7) ;
\draw [shift={(468.33,376.55)}, rotate = 0.23] [color={rgb, 255:red, 0; green, 0; blue, 0 }  ][line width=0.75]    (10.93,-3.29) .. controls (6.95,-1.4) and (3.31,-0.3) .. (0,0) .. controls (3.31,0.3) and (6.95,1.4) .. (10.93,3.29)   ;

\draw [line width=0.75]  [dash pattern={on 4.5pt off 4.5pt}]  (420.06,276.41) -- (448,276.83) ;
\draw [shift={(418.06,276.38)}, rotate = 0.85] [color={rgb, 255:red, 0; green, 0; blue, 0 }  ][line width=0.75]    (10.93,-3.29) .. controls (6.95,-1.4) and (3.31,-0.3) .. (0,0) .. controls (3.31,0.3) and (6.95,1.4) .. (10.93,3.29)   ;
\draw [line width=0.75]  [dash pattern={on 4.5pt off 4.5pt}]  (420.06,329.5) -- (448,329.92) ;
\draw [shift={(418.06,329.47)}, rotate = 0.85] [color={rgb, 255:red, 0; green, 0; blue, 0 }  ][line width=0.75]    (10.93,-3.29) .. controls (6.95,-1.4) and (3.31,-0.3) .. (0,0) .. controls (3.31,0.3) and (6.95,1.4) .. (10.93,3.29)   ;
\draw [line width=0.75]  [dash pattern={on 4.5pt off 4.5pt}]  (420.06,377.28) -- (448,377.7) ;
\draw [shift={(418.06,377.25)}, rotate = 0.85] [color={rgb, 255:red, 0; green, 0; blue, 0 }  ][line width=0.75]    (10.93,-3.29) .. controls (6.95,-1.4) and (3.31,-0.3) .. (0,0) .. controls (3.31,0.3) and (6.95,1.4) .. (10.93,3.29)   ;

\draw [line width=0.75]  [dash pattern={on 4.5pt off 4.5pt}]  (348.03,277.41) -- (377.5,277.83) ;
\draw [shift={(346.03,277.38)}, rotate = 0.81] [color={rgb, 255:red, 0; green, 0; blue, 0 }  ][line width=0.75]    (10.93,-3.29) .. controls (6.95,-1.4) and (3.31,-0.3) .. (0,0) .. controls (3.31,0.3) and (6.95,1.4) .. (10.93,3.29)   ;
\draw [line width=0.75]  [dash pattern={on 4.5pt off 4.5pt}]  (348.03,330.5) -- (377.5,330.92) ;
\draw [shift={(346.03,330.47)}, rotate = 0.81] [color={rgb, 255:red, 0; green, 0; blue, 0 }  ][line width=0.75]    (10.93,-3.29) .. controls (6.95,-1.4) and (3.31,-0.3) .. (0,0) .. controls (3.31,0.3) and (6.95,1.4) .. (10.93,3.29)   ;
\draw [line width=0.75]  [dash pattern={on 4.5pt off 4.5pt}]  (348.03,378.28) -- (377.5,378.7) ;
\draw [shift={(346.03,378.25)}, rotate = 0.81] [color={rgb, 255:red, 0; green, 0; blue, 0 }  ][line width=0.75]    (10.93,-3.29) .. controls (6.95,-1.4) and (3.31,-0.3) .. (0,0) .. controls (3.31,0.3) and (6.95,1.4) .. (10.93,3.29)   ;
\draw [line width=0.75]  [dash pattern={on 4.5pt off 4.5pt}]  (275.97,313.72) -- (293.04,280.17) ;
\draw [shift={(275.07,315.5)}, rotate = 296.96] [color={rgb, 255:red, 0; green, 0; blue, 0 }  ][line width=0.75]    (10.93,-3.29) .. controls (6.95,-1.4) and (3.31,-0.3) .. (0,0) .. controls (3.31,0.3) and (6.95,1.4) .. (10.93,3.29)   ;
\draw [line width=0.75]  [dash pattern={on 4.5pt off 4.5pt}]  (310.92,279.09) -- (293.04,280.17) ;
\draw [line width=0.75]  [dash pattern={on 4.5pt off 4.5pt}]  (310.75,381) -- (284.58,380.17) ;
\draw [line width=0.75]  [dash pattern={on 4.5pt off 4.5pt}]  (270.92,337.41) -- (284.58,380.17) ;
\draw [shift={(270.31,335.5)}, rotate = 72.28] [color={rgb, 255:red, 0; green, 0; blue, 0 }  ][line width=0.75]    (10.93,-3.29) .. controls (6.95,-1.4) and (3.31,-0.3) .. (0,0) .. controls (3.31,0.3) and (6.95,1.4) .. (10.93,3.29)   ;
\draw [line width=0.75]  [dash pattern={on 4.5pt off 4.5pt}]  (279.97,330.84) -- (311.54,330.88) ;
\draw [shift={(277.97,330.83)}, rotate = 0.09] [color={rgb, 255:red, 0; green, 0; blue, 0 }  ][line width=0.75]    (10.93,-3.29) .. controls (6.95,-1.4) and (3.31,-0.3) .. (0,0) .. controls (3.31,0.3) and (6.95,1.4) .. (10.93,3.29)   ;
\draw  [line width=0.75]  (423.7,171.54) .. controls (423.7,166.73) and (428.09,162.83) .. (433.51,162.83) .. controls (438.94,162.83) and (443.33,166.73) .. (443.33,171.54) .. controls (443.33,176.35) and (438.94,180.25) .. (433.51,180.25) .. controls (428.09,180.25) and (423.7,176.35) .. (423.7,171.54) -- cycle ; \draw  [line width=0.75]  (423.7,171.54) -- (443.33,171.54) ; \draw  [line width=0.75]  (433.51,162.83) -- (433.51,180.25) ;
\draw  [line width=0.75]  (461.7,206.54) .. controls (461.7,201.73) and (466.09,197.83) .. (471.51,197.83) .. controls (476.94,197.83) and (481.33,201.73) .. (481.33,206.54) .. controls (481.33,211.35) and (476.94,215.25) .. (471.51,215.25) .. controls (466.09,215.25) and (461.7,211.35) .. (461.7,206.54) -- cycle ; \draw  [line width=0.75]  (461.7,206.54) -- (481.33,206.54) ; \draw  [line width=0.75]  (471.51,197.83) -- (471.51,215.25) ;
\draw  [line width=0.75]  (448.45,272.95) .. controls (448.45,268.14) and (452.85,264.24) .. (458.27,264.24) .. controls (463.69,264.24) and (468.09,268.14) .. (468.09,272.95) .. controls (468.09,277.76) and (463.69,281.66) .. (458.27,281.66) .. controls (452.85,281.66) and (448.45,277.76) .. (448.45,272.95) -- cycle ; \draw  [line width=0.75]  (448.45,272.95) -- (468.09,272.95) ; \draw  [line width=0.75]  (458.27,264.24) -- (458.27,281.66) ;
\draw  [line width=0.75]  (449.45,324.95) .. controls (449.45,320.14) and (453.85,316.24) .. (459.27,316.24) .. controls (464.69,316.24) and (469.09,320.14) .. (469.09,324.95) .. controls (469.09,329.76) and (464.69,333.66) .. (459.27,333.66) .. controls (453.85,333.66) and (449.45,329.76) .. (449.45,324.95) -- cycle ; \draw  [line width=0.75]  (449.45,324.95) -- (469.09,324.95) ; \draw  [line width=0.75]  (459.27,316.24) -- (459.27,333.66) ;
\draw  [line width=0.75]  (449,372.7) .. controls (449,367.89) and (453.4,363.99) .. (458.82,363.99) .. controls (464.24,363.99) and (468.64,367.89) .. (468.64,372.7) .. controls (468.64,377.51) and (464.24,381.41) .. (458.82,381.41) .. controls (453.4,381.41) and (449,377.51) .. (449,372.7) -- cycle ; \draw  [line width=0.75]  (449,372.7) -- (468.64,372.7) ; \draw  [line width=0.75]  (458.82,363.99) -- (458.82,381.41) ;

\draw (398.79,179.84) node  [font=\huge] [align=left] {{\fontfamily{pcr}\selectfont {\large .}}};
\draw (398.79,187.32) node  [font=\huge] [align=left] {{\fontfamily{pcr}\selectfont {\large .}}};
\draw (398.79,195.18) node  [font=\huge] [align=left] {{\fontfamily{pcr}\selectfont {\large .}}};
\draw (561.54,308.44) node  [font=\huge] [align=left] {{\fontfamily{pcr}\selectfont {\large .}}};
\draw (561.54,300.58) node  [font=\huge] [align=left] {{\fontfamily{pcr}\selectfont {\large .}}};
\draw (561.54,293.1) node  [font=\huge] [align=left] {{\fontfamily{pcr}\selectfont {\large .}}};
\draw (521.47,270.32) node  [font=\footnotesize]  {$( \cdot )^{2}$};
\draw (523.03,324.81) node  [font=\footnotesize]  {$( \cdot )^{2}$};
\draw (521.47,372.38) node  [font=\footnotesize]  {$( \cdot )^{2}$};
\draw (185.32,324.26) node  [font=\huge,rotate=-269.79] [align=left] {{\fontfamily{pcr}\selectfont {\large .}}};
\draw (176.72,324.26) node  [font=\huge,rotate=-269.79] [align=left] {{\fontfamily{pcr}\selectfont {\large .}}};
\draw (167.26,324.26) node  [font=\huge,rotate=-269.79] [align=left] {{\fontfamily{pcr}\selectfont {\large .}}};
\draw (328.75,327.11) node    {$\sum $};
\draw (328.75,377.02) node    {$\sum $};
\draw (328.75,270.44) node    {$\sum $};
\draw (215.12,323.56) node  [font=\scriptsize]  {$\ z[ n_{s} -1]$};
\draw (132.44,323.69) node  [font=\scriptsize,rotate=-0.07]  {$\ z[ -n_{s} +1]$};
\draw (96.73,330.98) node  [font=\huge] [align=left] {{\fontfamily{pcr}\selectfont {\large .}}};
\draw (96.73,323.12) node  [font=\huge] [align=left] {{\fontfamily{pcr}\selectfont {\large .}}};
\draw (96.73,315.64) node  [font=\huge] [align=left] {{\fontfamily{pcr}\selectfont {\large .}}};
\draw (95.68,188.42) node  [font=\scriptsize,rotate=-0.07]  {$\ y[ 0]$};
\draw (127.95,187.39) node  [font=\scriptsize,rotate=-358.27]  {$\ y[ 1]$};
\draw (200.43,187.87) node  [font=\scriptsize]  {$\ y[ Nn_{s} -1]$};
\draw (429.82,65.36) node   [align=left] {{\fontfamily{ptm}\selectfont Offline Exponential Generator}};
\draw (293.89,127.78) node  [font=\scriptsize]  {$\overbrace{( y[ 0] ,\ y[ 0] ,\ ...\ ,\ y[ 0]}^{L})$};
\draw (342.99,224.65) node  [font=\scriptsize]  {$( y[ Nn_{s} -1] ,\ y[ Nn_{s} -1] ,\ ...\ ,\ y[ Nn_{s} -1])$};
\draw (479.55,305.01) node  [font=\scriptsize]  {$-1/\sigma ^{2}_{\rm{w}}$};
\draw (479.55,352.79) node  [font=\scriptsize]  {$-1/\sigma ^{2}_{\rm{w}}$};
\draw (479.55,251.92) node  [font=\scriptsize]  {$-1/\sigma ^{2}_{\rm{w}}$};
\draw (70.45,324.55) node  [rotate=-270] [align=left] {{\fontfamily{ptm}\selectfont Argmax ( . )}};
\draw (7.73,317.4) node    {$\hat{d}$};
\draw (440.54,109.69) node  [font=\huge,rotate=-90] [align=left] {{\fontfamily{pcr}\selectfont {\large .}}};
\draw (448.74,109.69) node  [font=\huge,rotate=-90] [align=left] {{\fontfamily{pcr}\selectfont {\large .}}};
\draw (456.55,109.69) node  [font=\huge,rotate=-90] [align=left] {{\fontfamily{pcr}\selectfont {\large .}}};
\draw (149.54,188.69) node  [font=\huge,rotate=-90] [align=left] {{\fontfamily{pcr}\selectfont {\large .}}};
\draw (157.74,188.69) node  [font=\huge,rotate=-90] [align=left] {{\fontfamily{pcr}\selectfont {\large .}}};
\draw (165.55,188.69) node  [font=\huge,rotate=-90] [align=left] {{\fontfamily{pcr}\selectfont {\large .}}};
\draw (410,135) node    {$-$};
\draw (448.67,159.67) node    {$-$};
\draw (485.33,193.33) node    {$-$};
\draw (397.76,271.64) node  [font=\footnotesize]  {${\rm{exp}}( \cdot )$};
\draw (397.76,326.7) node  [font=\footnotesize]  {${\rm{exp}}( \cdot )$};
\draw (397.76,373.2) node  [font=\footnotesize]  {${\rm{exp}}( \cdot )$};
\draw (319.82,100.06) node  [font=\scriptsize]  {$( v_{0}[ 0] ,\ v_{0}[ 1] ,\ ...\ ,\ v_{0}[ L-1])$};
\draw (579.82,100.06) node  [font=\scriptsize]  {$( v_{Nn_{s} -1}[ 0] ,\ v_{Nn_{s} -1}[ 1] ,\ ...\ ,\ v_{Nn_{s} -1}[ L-1])$};
\draw (310.32,184.98) node  [font=\scriptsize]  {$( y[ 1] ,\ y[ 1] ,\ ...\ ,\ y[ 1])$};

\end{tikzpicture}

  \caption{MCS implementation of the proposed theoretical NDA-ML TO estimator in \eqref{8989ioio} using \eqref{eq: monte c}. The solid and dashed lines represent the in-phase and quadrature components of the received samples, respectively. }\label{fig: gamma block}
\end{figure*}

\subsection{Iterative Likelihood Maximization}
Efficient one dimensional
iterative search algorithms can be used to avoid exhaustive search, and thus, reduce the computational complexity of \eqref{8989ioio}. 
In iterative search method, an interval $[d_{\rm{L}}, d_{\rm{U}}]$ containing the true \ac{to} $d^*$ is
established and is then repeatedly reduced on the basis of function evaluations
until a reduced bracket $[d_{\rm{L}}, d_{\rm{U}}]$ is achieved which is sufficiently small. The
minimizer/maximizer can be assumed to be at the center of interval $[d_{\rm{L}}, d_{\rm{U}}]$. These
methods can be applied to any function and differentiability of the function is not
essential.

An iterative search method in which iterations can be performed
until the desired accuracy in either the maximizer or the maximum value of the
objective function is achieved is the golden-section search method \cite{press2007numerical}. 
For a strictly unimodal function with an extremum inside the interval, Golden-section search method finds that extremum, while for an interval containing multiple extrema (possibly including the interval boundaries), it converges to one of them. Implementation of the proposed NDA-ML TO estimation with golden-section search is summarized in Algorithm
\ref{Table1xrrrr}. For Algorithm \ref{Table1xrrrr}, we define
\begin{equation}
\mathcal{L}(d) \triangleq 
{\rm{Prod}}\Big(f^{(d:d+Nn_{\rm{s}}-1)}_{{\bf Y}}({\bf y}; {\rm{H}}_0)\Big),    \end{equation}
where ${\rm{Prod}}\big{(}[z_0,z_1,\dots,z_{u-1}]\big{)} \triangleq \prod_{i=0}^{u-1}z_i$, 
and
\begin{align} \label{6yo09}
 &   {{\bf f}^{(d:q)}_{\bf{Y}}} (\cdot ; {\rm{H}_0} )\triangleq
\\ \nonumber    
&\big{[}f_{Y[d]}(\cdot | {\rm{H}_0} )~f_{Y[d+1]}(\cdot | {\rm{H}_0} )~ \dots~ f_{Y[q]}(\cdot | {\rm{H}_0} )\big{]}^T
\end{align}
with $q \ge d$ and $f_{Y[nn_{\rm{s}}+d]}(\cdot | {\rm{H}_0})$ as the \ac{pdf} of the received sample $y[nn_{\rm{s}}+d]  \triangleq y_{n}[d]$ given ${\rm{H}_0}$.

\begin{algorithm}[t]
    \caption{Golden-section search}\label{Table1xrrrr}
    \begin{algorithmic}[1]
      \Statex \textbf{Initialization:} $\mathcal{D} \gets \{ -n_{\rm{s}},\dots, n_{\rm{s}}-2\}$, $init \gets -n_{\rm{s}}+1$, $last \gets n_{\rm{s}}-1$, $ratio \gets 0.381966$ 
      \State $c  \gets init + +\lfloor ratio * (last-init)\rfloor$
      \While{$\mathcal{L}(c)<\mathcal{L}(init)$ or $\mathcal{L}(c)<\mathcal{L}(last)$}
      \If{$|\mathcal{D}|$=2}
      \If{$\mathcal{L}(init)>\mathcal{L}((last)$}
      \State \Return $\hat{d}^{\textrm{opt}} \gets init$
      \Else 
      \State \Return $\hat{d}^{\textrm{opt}} \gets last$
      \EndIf
      \EndIf
      \State $\mathcal{D} \gets \mathcal{D}\setminus \{c\}$ 
      \State $c \gets$ Choose a random index from the set 
       $\mathcal{D}$ 
      \EndWhile
      
      \While{$last-init \ge 4$}
      \If{$last-c \ge c-init$}
          \State $d \gets c+\lfloor ratio * (last-init)\rfloor$
          \If{$\mathcal{L}(d)<\mathcal{L}(c)$} 
            \State $last \gets d$
          \Else
            \State $init \gets c$
            \State $c \gets d$
          \EndIf
      \Else 
          \State $d \gets c-\lfloor ratio*(last-init) \rfloor$
          \If{$\mathcal{L}(d)<\mathcal{L}(c)$} 
            \State $init \gets d$
          \Else
            \State $last \gets c$
            \State $c \gets d$
          \EndIf
      
      \EndIf
      \EndWhile
      \State $\hat{d}^{\textrm{opt}} \gets \operatorname*{argmax}_{i \in \{init,\dots,last\}}  \mathcal{L}(i)$ 
      \State \Return $\hat{d}^{\textrm{opt}}$
    \end{algorithmic}
  \end{algorithm}

\subsection{Complexity Analysis}

Here, we compare the complexity of the proposed theoretical and the MCS time synchronization algorithms. In our complexity analysis, a real addition, multiplication, or division is counted as one floating point operation (FLOP). Considering the fact that the number of FLOPs for $H_{\rm{d}}$, $d \ge 0$, is higher than that for $d <0$, the number of FLOPs per TO hypothesis for the theoretical and the MCS time synchronization algorithms are upper bounded by  $2N(n_{\rm{x}}+n_{\rm{z}})+(2N-1)n_{\rm{x}}n_{\rm{h}}(8+2u_1+2u_2)$ and $((6+2u_1)(L-1)+1)(n_{\rm{x}}+n_{\rm{z}}-1)(N-1)$, respectively, where $u_1$ and $u_2$ denote the number of FLOPs for the computation of $\exp(\cdot)$ and error function $\Phi(\cdot)$, respectively. Here, $L$ is a trade-off parameter between the accuracy and complexity for the MCS algorithm. That is, increasing $L$ increases both accuracy and complexity of the MCS algorithm.




The ratio of the average estimation time (RAET) versus $n_{\rm{x}}$ for the proposed theoretical NDA-ML and the MCS algorithms are shown in Table \ref{table: time analysis}.  
Here, RAET is defined as:
\begin{equation}
 \text{RAET}= \frac{\text{Average estimation time for MCS}}{\text{Average estimation time for the theoretical}} .  
\end{equation}

The number of  Monte Carlo samples is set to $L=10^4$ for the MCS implementation, and exhaustive search method is used for both algorithms.
As seen, MCS implementation offers a lower computational complexity compared to the theoretical \ac{ml} estimator using  Theorem \ref{theo: pdf y}. This complexity reduction is obtained at the expense of an insignificant performance degradation in terms of lock-in probability as it will be shown in the  next section.


\begin{table}[t!]
\vspace{-1em}
\centering 
 \caption{Complexity Analysis}
\label{table: time analysis}
\resizebox{0.4\textwidth}{!}{
\begin{tabularx}{0.3\textwidth}{@{  }l*{4}{C}c@{}}
\toprule
$n_{\rm{x}}$ & 64    & 128 & 256 & 512 & 1024 \\ 
\midrule
RAET  &  0.841   & 0.779 & 0.704 & 0.668 & 0.629   \\ 
\bottomrule
\end{tabularx} }
\end{table}

\section{Simulations}\label{simmp}
In this section, we evaluate the performance of the proposed \ac{nda}-\ac{ml} time synchronization algorithm through several simulation experiments. 

\subsection{Simulation Setup}
We consider a \ac{zp}-OFDM system with 128-QAM modulation in a frequency-selective Rayleigh fading channel. 
Unless otherwise mentioned, the number of sub-carriers is $n_{\rm{x}}=128$, the number of zero-padded samples is $n_{\rm{z}}=15$, and
 the number of observed OFDM symbols at the receiver is $N=10$.
The sampling time of the ZP-OFDM  system at the receiver is $T_{\rm sa}=10^{-6}$s.
An uncorrelated multipath fading channel with $n_{\rm{h}}=10$ taps and maximum delay spread of {$\tau_{\rm{max}}=10 \mu$s} is considered. 
The  delay  profile  of the  Rayleigh
fading  channel in \eqref{7u8i0000} is modeled as an exponential-decay function, i.e., $\sigma^2_{h_l}=\alpha\exp(-\beta l)$, $l=0,1,\dots,n_{\rm{h}}-1$, where  $p_{\rm{h}}=\sum_{l=0}^{n_{\rm{h}}-1}\sigma^2_{h_l}=1$,  $\alpha=1/2.5244$, and $\beta=0.5$. The maximum Doppler spread of the fading channel  is set to $f_{\rm{D}}=5$ Hz.
Without loss of generality, the transmit power is assumed to be {$\sigma_{\rm{x}}^2=1$}, and  
the \ac{awgn} is modeled as  a zero-mean complex Gaussian random variable with variance $\sigma^2_{\rm{w}}$, which varies according to the value of \ac{snr} $\gamma \triangleq \sigma_{\rm x}^2 p_{\rm h} / \sigma_{\rm w}^2$. The \ac{to} introduced to the system is modeled as a uniformly distributed integer random variable in the range of $d \in [-30 ,  \ 30]$.  Simulations are evaluated under $10^4$ Monte Carlo realizations, and the number of samples for MCS implementation of the proposed theoretical NDA-ML algorithm is $L=10^4$. The performance of the proposed algorithms are evaluated in terms of mean squared error (MSE) and lock-in probability. 
Here, the lock-in probability is defined as the probability that the estimated TO (given in sampling time) equals to the actual TO. That is, any non-zero error is counted as a missed estimation.

\subsection{Simulation Results}
The performance of the proposed theoretical NDA-ML algorithm, its MCS implementation, and the current state-of-the-art NDA \ac{to} estimator for \ac{zp}-OFDM, i.e. transition metric  (TM) \cite{LeNir2010}, for different values of $E_{\rm{b}}/N_0$ are shown in Fig.~\ref{fig: snr}.  As can be seen, the proposed theoretical algorithm and its  MCS implementation outperform  the TM  algorithm since they maximize the likelihood function while TM is a heuristic algorithm.  Moreover, as seen, there is a negligible performance gap between the proposed theoretical NDA-ML algorithm and its MCS implementation. 
This performance gap can be further reduced by increasing the number of Monte  Carlo samples $L$ used for averaging in \eqref{eq: monte c} at the expense of higher complexity. In Fig.~\ref{fig: snr}. we also illustrate the performance of the  the sub-optimal time synchronization  algorithm  in  \cite{K2020},  which relies  on  Gaussian  PDF  approximation  of  the  received samples.  As can be seen, there is a large gap between the proposed algorithms and the sub-optimal algorithm in \cite{K2020} at low SNR values.

Fig. \ref{fig: dop}  illustrates the effect of the maximum Doppler spread (mobility) on the performance of the proposed theoretical NDA-ML algorithm and its MCS implementation. As seen, the lock-in probability increases as the maximum Doppler spread increases. The reason is that the time dynamics of 
the channel taps contributing (through convolution) to the received samples become less correlated as the maximum Doppler spread increases. 
However, for zero maximum Doppler spread, identical channel taps contribute to the received samples.
Thus, our independency assumption on the received samples becomes more valid for higher values of maximum Doppler spread. These results reveal that the proposed NDA-ML algorithm can be considered as a promising candidate for vehicle-to-vehicle (V2V) communications.

The effect of the number of OFDM symbols $N$, used for time synchronization, on the performance of the 
proposed theoretical NDA-ML algorithm, its MCS implementation, and the TM algorithm \cite{LeNir2010} are represented in Fig. \ref{fig: obs}. As expected, the higher $N$, the higher the lock-in probability. 
Major improvements in performance occurs when the number of OFDM symbols increases from 1 to 10, and then the rate of performance improvement decreases. This is due to the fact that innovation introduced by each new sample to an ML estimator deacreases as the total number of samples (used for estimation) increases.

In Fig. \ref{fig: tap}, the performance of the proposed theoretical NDA-ML TO estimator, its MCS implementation, and the TM estimator \cite{LeNir2010} versus the number of channel taps $n_{\rm{h}}$ for $n_{\rm{z}}=20$ at $15$ dB $E_{\rm{b}}/N_0$ are shown. 
As seen, the lock-in  probability of the theoretical NDA-ML TO estimator and its MCS implementation degrades as $n_{\rm{h}}$ increases. This is because the sharpness of the likelihood function decreases; the sharpness of likelihood function determines how accurately we can estimate an unknown parameter.

\begin{figure}
\vspace{-1em}
\centering
\includegraphics[height=2.835in]{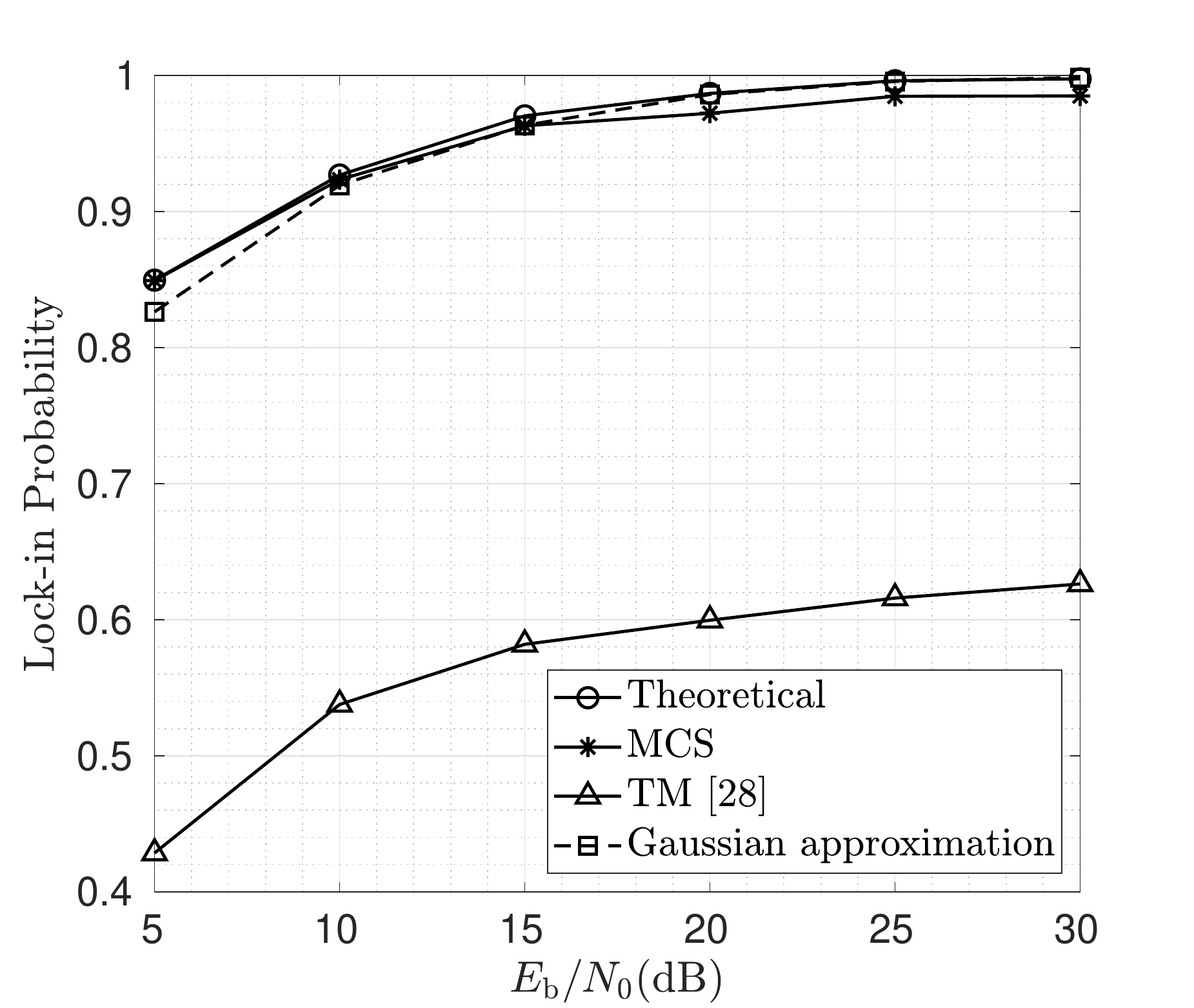}
  \caption{Lock-in probability versus $E_{\rm{b}}/N_0$.}\label{fig: snr}
\end{figure}

\begin{figure}
\vspace{-2.1em}
\centering
\includegraphics[height=2.835in]{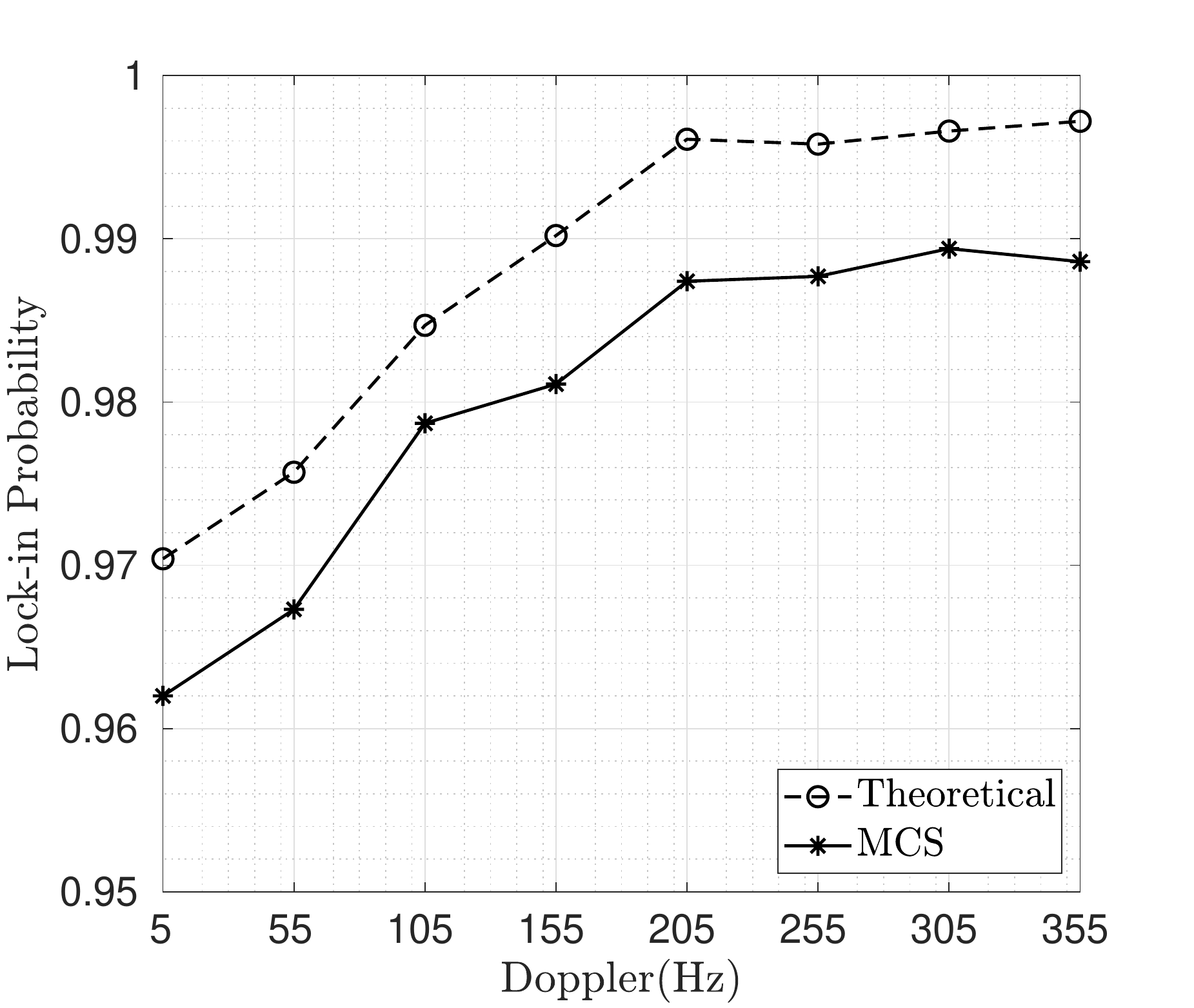}
  \caption{Lock-in probability versus maximum Doppler spread of the fading channel at $15$ dB $E_{\rm{b}}/N_0$.} \label{fig: dop}
  \vspace{-1em}
\end{figure}

\begin{figure}
\centering
\includegraphics[height=2.835in]{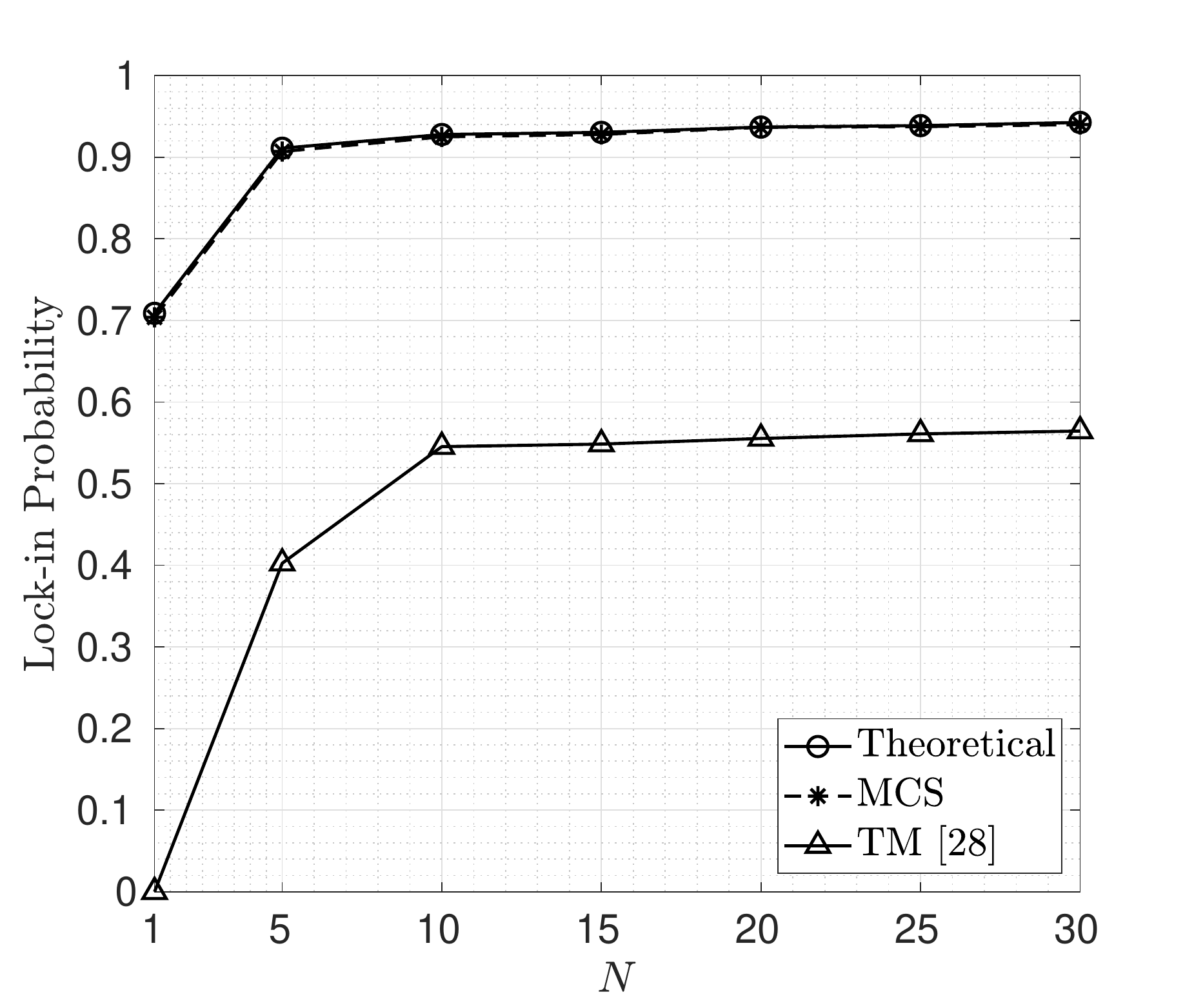}
  \caption{Lock-in probability versus the number of observation vectors $N$ at $10$ dB $E_{\rm{b}}/N_0$.}\label{fig: obs}
\end{figure}
\begin{figure}
\vspace{-2em}
\centering
\includegraphics[height=2.835in]{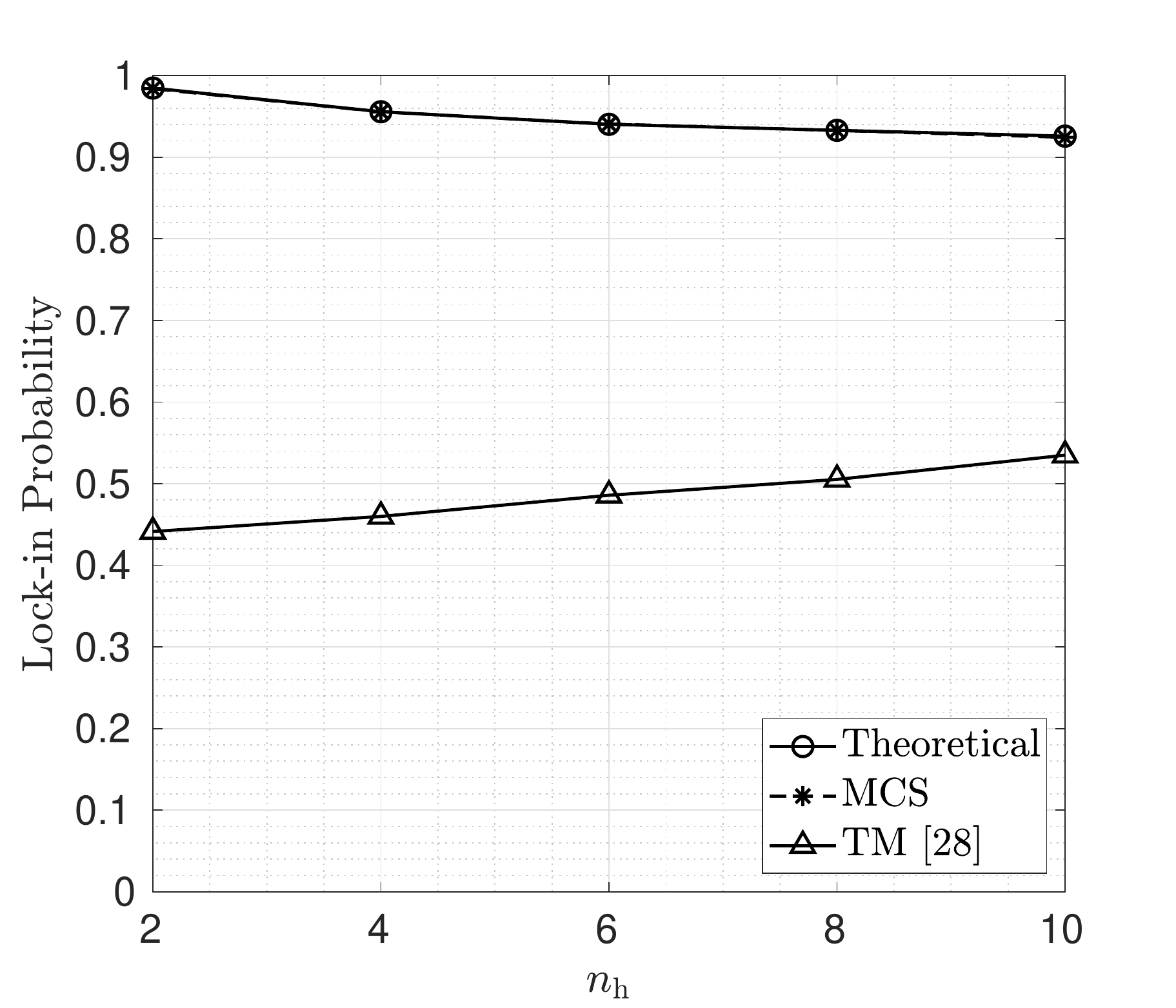}
  \caption{Lock-in probability versus the number of channel taps $n_{\rm{h}}$ at $10$ dB $E_{\rm{b}}/N_0$.}\label{fig: tap}
\end{figure}

In Fig.~\ref{fig: hist}, we  illustrate the empirical probability mass function (PMF) of the synchronization error for the proposed theoretical and the MCS algorithms at 10 dB $E_{\rm{b}}/N_0$. As can be seen, the empirical PMF of the error is not symmetric around zero and is slightly biased towards positive \ac{to}s. Based on asymptotic properties of the MLEs, the biased term approaches zero as $N\rightarrow \infty$. Moreover, we observe that the synchronization error falls in small interval, i.e., $\{-2,-1,1,2\}$. This means that the proposed algorithms offer low 
MSE as shown in Fig. \ref{fig: mse}. 
Because of low MSE, the proposed  time-synchronization algorithms can take advantage of low complexity channel coding to further improve synchronization performance.

The effect of PDP estimation error on the performance of the proposed theoretical and the MCS time synchronization algorithms is shown in Fig.~\ref{fig: sens}. We model the estimated PDP as 
\begin{align}
\hat{\sigma}^2_{{\rm{h}}_k} \in {\cal{U}}\Big{[}\sigma^2_{{\rm{h}}_k}-\alpha\sigma^2_{{\rm{h}}_k},\sigma^2_{{\rm{h}}_k}+\alpha\sigma^2_{{\rm{h}}_k}\Big{]},
\end{align}
where $\sigma^2_{{\rm{h}}_k}$, $k=0,1,\dots, n_{\rm{h}}-1$, is the true PDP, and  
${\cal{U}}[a,b]$ denotes the uniform distribution in the interval $[a,b]$. 
In Fig.~\ref{fig: sens}, we show the lock-in probability versus $\alpha \in [0,1]$ at 10 dB  $E_{\rm{b}}/N_0$. As can be seen, the theoretical and the MCS algorithms are robust to the 
the delay profile estimation error for $\alpha \in [0,1]$ and $\alpha \in [0,0.5]$, respectively. 
While the performance of the theoretical algorithm slightly degrades for
$\alpha \in [0.5,1]$, the lock-in probability is still larger than $0.75$.

\begin{figure}
\centering
\includegraphics[height=2.835in]{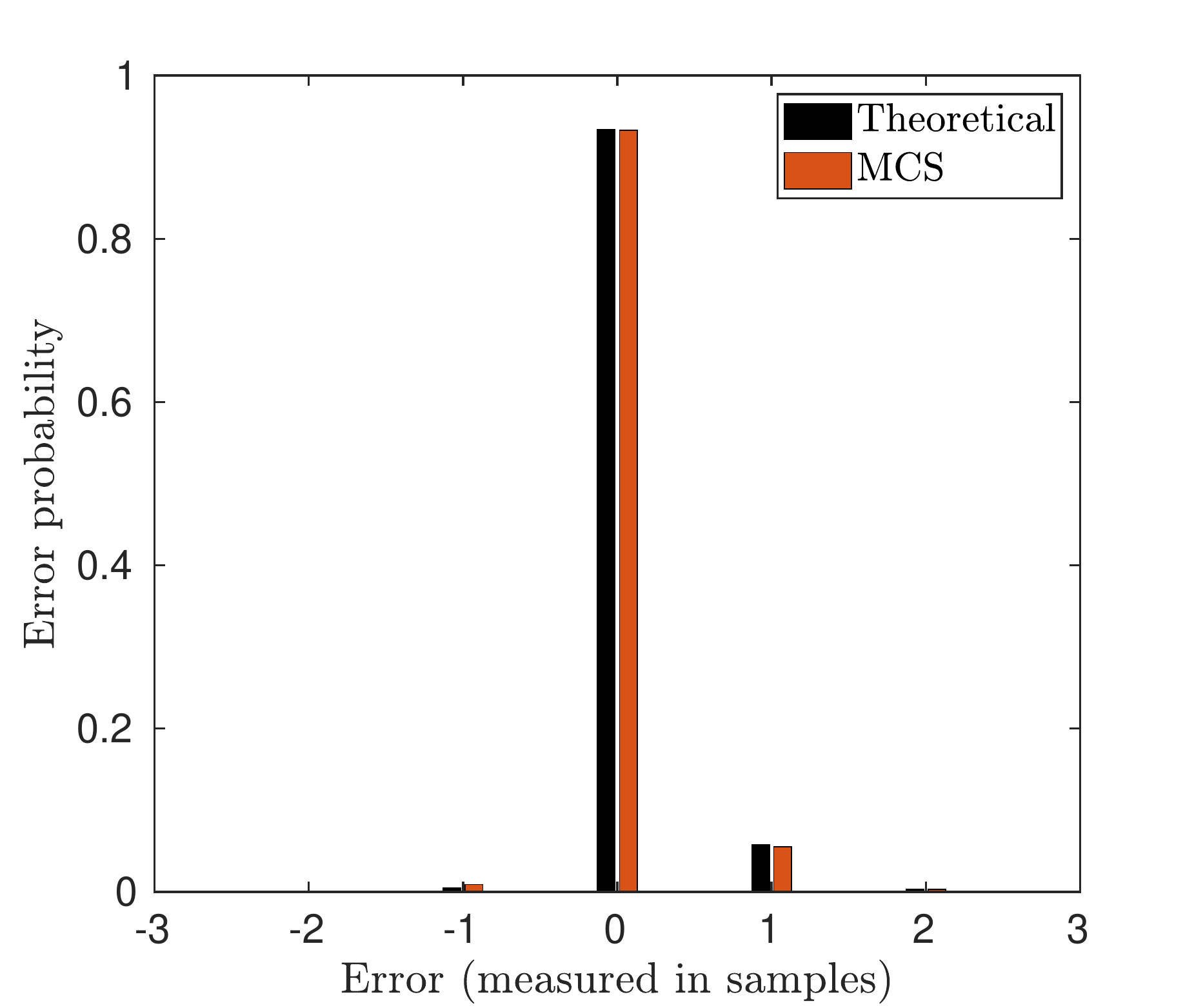}
  \caption{PMF of the synchronization error for the theoretical and the MCS algorithms at 10 dB  $E_{\rm{b}}/N_0$.}\label{fig: hist}
\end{figure}

\begin{figure}
\centering
\vspace{-2em}
\includegraphics[height=2.835in]{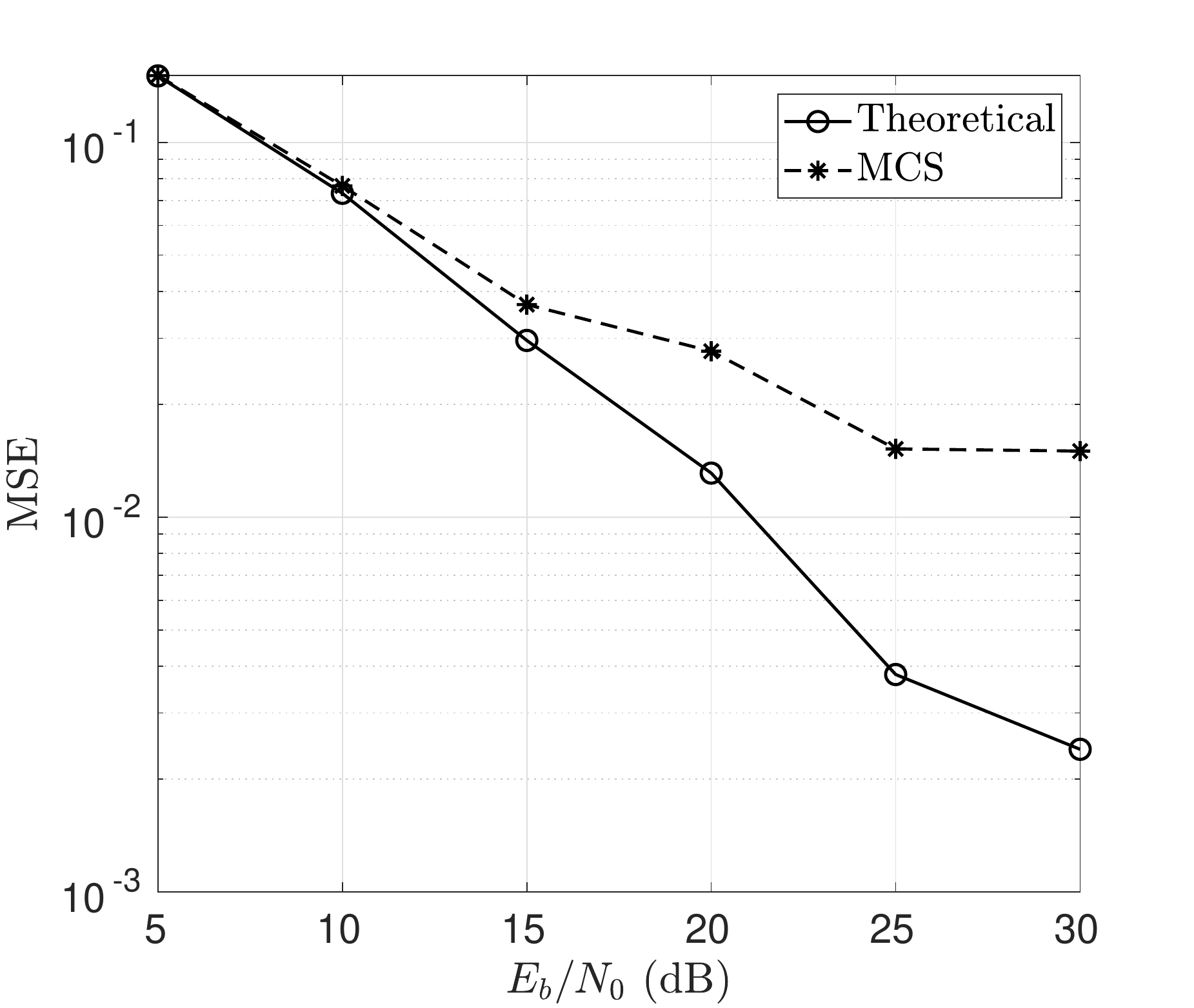}
  \caption{Mean Squared Error (MSE) versus SNR for Theoretical and MCS method.}\label{fig: mse}
\end{figure}

\begin{figure}
\centering
\includegraphics[height=2.835in]{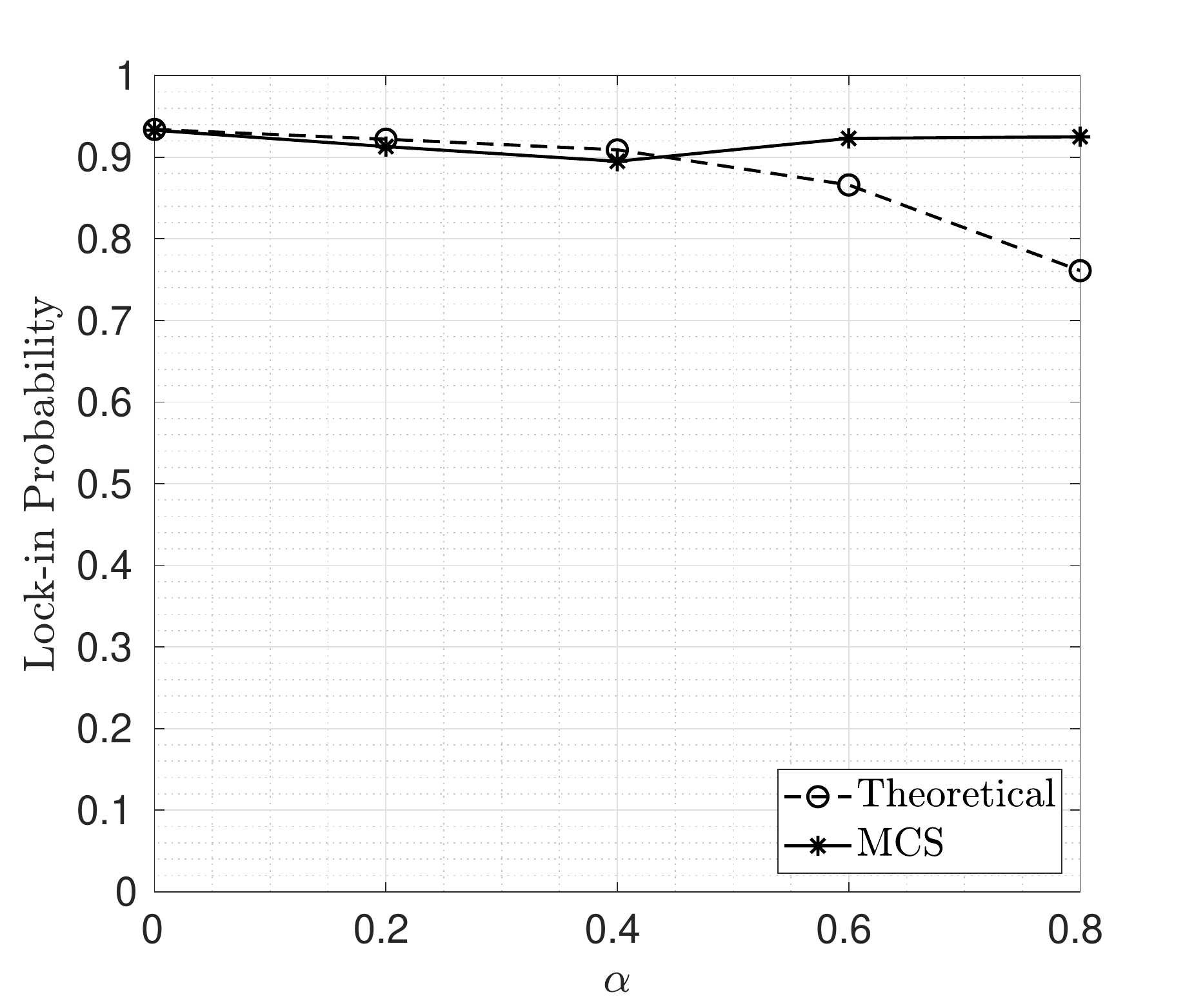}
  \caption{Sensitivity of the proposed time synchronization algorithms to PDP estimation error for $N=10$ at 10 dB  $E_{\rm{b}}/N_0$.}\label{fig: sens}
\end{figure}

\section{Conclusions}\label{sec: conclu}
In this paper, for the first time in the literature, the \ac{nda}-\ac{ml} time synchronization for \ac{zp}-OFDM was analytically derived and a feasible solution for its implementation based on MCS technique was proposed.
The obtained time synchronization method can be employed for both frame and symbol synchronization. Moreover,  
to achieve the optimal time synchronization method, we obtained a closed-form approximate expression for the distribution of convolution, i.e. received convolved signal.
Simulation results verify that the proposed theoretical \ac{nda}-\ac{ml} time synchronization and its MCS implementation can offer high lock-in probabilities even at low SNR values. Also, they are effective in highly time-selective channels with large maximum Doppler spread. 
These properties make \ac{zp}-OFDM a promising candidate for low-power IoT networks and V2V communications.

\IEEEpeerreviewmaketitle

\bibliographystyle{IEEEtran} 
\bibliography{IEEEabrv, references.bib}

\begin{thebibliography}{10}
\providecommand{\url}[1]{#1}
\csname url@samestyle\endcsname
\providecommand{\newblock}{\relax}
\providecommand{\bibinfo}[2]{#2}
\providecommand{\BIBentrySTDinterwordspacing}{\spaceskip=0pt\relax}
\providecommand{\BIBentryALTinterwordstretchfactor}{4}
\providecommand{\BIBentryALTinterwordspacing}{\spaceskip=\fontdimen2\font plus
\BIBentryALTinterwordstretchfactor\fontdimen3\font minus
  \fontdimen4\font\relax}
\providecommand{\BIBforeignlanguage}[2]{{%
\expandafter\ifx\csname l@#1\endcsname\relax
\typeout{** WARNING: IEEEtran.bst: No hyphenation pattern has been}%
\typeout{** loaded for the language `#1'. Using the pattern for}%
\typeout{** the default language instead.}%
\else
\language=\csname l@#1\endcsname
\fi
#2}}
\providecommand{\BIBdecl}{\relax}
\BIBdecl

\bibitem{farhang2016ofdm}
B.~Farhang-Boroujeny and H.~Moradi, ``{OFDM} inspired waveforms for {5G},''
  \emph{{IEEE} Commun. Surveys Tuts.}, vol.~18, no.~4, pp. 2474--2492, 4Q.
  2016.

\bibitem{zhang2019fine}
X.~Zhang, L.~Zhang, P.~Xiao, and J.~Wei, ``Fine timing synchronization based on
  modified expectation maximization clustering algorithm for {OFDM} systems,''
  \emph{{IEEE} Commun. Lett.}, vol.~8, no.~5, pp. 1452--1455, Oct. 2019.

\bibitem{abdzadeh2019timing}
H.~Abdzadeh-Ziabari, W.-P. Zhu, and M.~Swamy, ``Timing and frequency
  synchronization and doubly selective channel estimation for {OFDMA} uplink,''
  \emph{{IEEE} Trans. Circuits Syst. {II}}, vol.~67, no.~1, pp. 62--66, Jan.
  2020.

\bibitem{zhang2015maximum}
X.~Zhang, J.~Liu, H.~Li, and B.~Himed, ``Maximum likelihood synchronization for
  {DVB-T2} in unknown fading channels,'' \emph{{IEEE} Trans. Broadcast.},
  vol.~61, no.~4, pp. 615--624, Dec. 2015.

\bibitem{mohebbi2014novel}
A.~Mohebbi, H.~Abdzadeh-Ziabari, and M.~G. Shayesteh, ``Novel coarse timing
  synchronization methods in {OFDM} systems using fourth-order statistics,''
  \emph{{IEEE} Trans. Veh. Technol.}, vol.~64, no.~5, pp. 1904--1917, May 2014.

\bibitem{morelli2007synchronization}
M.~Morelli, C.-C.~J. Kuo, and M.-O. Pun, ``Synchronization techniques for
  orthogonal frequency division multiple access {(OFDMA)}: A tutorial review,''
  \emph{Proc. {IEEE}}, vol.~95, no.~7, pp. 1394--1427, Jul. 2007.

\bibitem{park2004blind}
B.~Park, H.~Cheon, E.~Ko, C.~Kang, and D.~Hong, ``A blind {OFDM}
  synchronization algorithm based on cyclic correlation,'' \emph{{IEEE} Signal
  Process. Lett.}, vol.~11, no.~2, pp. 83--85, Feb. 2004.

\bibitem{lin2018analysis}
D.~W. Lin, ``An analysis of the performance of {ML} blind {OFDM} symbol timing
  estimation,'' \emph{{IEEE} Trans. Signal Process.}, vol.~66, no.~20, pp.
  5324--5337, Oct. 2018.

\bibitem{ziamaxli2018}
H.~{Abdzadeh-Ziabari}, W.~{Zhu}, and M.~N.~S. {Swamy}, ``Joint maximum
  likelihood timing, frequency offset, and doubly selective channel estimation
  for {OFDM} systems,'' \emph{{IEEE} Trans. Veh. Technol.}, vol.~67, no.~3, pp.
  2787--2791, Mar. 2018.

\bibitem{abdzadeh2016improved}
H.~Abdzadeh-Ziabari, W.-P. Zhu, and M.~Swamy, ``Improved coarse timing
  estimation in {OFDM} systems using high-order statistics,'' \emph{{IEEE}
  Trans. Commun.}, vol.~64, no.~12, pp. 5239--5253, 2016.

\bibitem{gul2014timing}
M.~M.~U. Gul, X.~Ma, and S.~Lee, ``Timing and frequency synchronization for
  {OFDM} downlink transmissions using {Zadoff-Chu} sequences,'' \emph{{IEEE}
  Trans. Wireless Commun.}, vol.~14, no.~3, pp. 1716--1729, Mar. 2015.

\bibitem{nasir2016timing}
A.~A. Nasir, S.~Durrani, H.~Mehrpouyan, S.~D. Blostein, and R.~A. Kennedy,
  ``Timing and carrier synchronization in wireless communication systems: a
  survey and classification of research in the last 5 years,'' \emph{EURASIP
  Journal on Wireless Communications and Networking}, vol. 2016, no.~1, p. 180,
  Aug. 2016.

\bibitem{de20195g}
I.~B.~F. de~Almeida, L.~L. Mendes, J.~J. Rodrigues, and M.~A. da~Cruz, ``{5G}
  waveforms for {IoT} applications,'' \emph{{IEEE} Commun. Surveys Tuts.},
  vol.~21, no.~3, pp. 2554--2567, 3Q. 2019.

\bibitem{ven5G}
S.~{Venkatesan} and R.~A. {Valenzuela}, ``{OFDM} for {5G}: Cyclic prefix versus
  zero postfix, and filtering versus windowing,'' in \emph{Proc. {IEEE} {ICC}},
  Kuala Lumpur, Malaysia, May 2016, pp. 1--5.

\bibitem{wang2015maximum}
P.-S. Wang and D.~W. Lin, ``On maximum-likelihood blind synchronization over
  {WSSUS} channels for {OFDM} systems,'' \emph{{IEEE} Trans. Signal Process.},
  vol.~63, no.~19, pp. 5045--5059, Oct. 2015.

\bibitem{chin2011blind}
W.-L. Chin, ``Blind symbol synchronization for {OFDM} systems using cyclic
  prefix in time-variant and long-echo fading channels,'' \emph{{IEEE} Trans.
  Veh. Technol.}, vol.~61, no.~1, pp. 185--195, Jan. 2012.

\bibitem{van1997ml}
J.-J. Van~de Beek, M.~Sandell, and P.~O. Borjesson, ``{ML} estimation of time
  and frequency offset in {OFDM} systems,'' \emph{{IEEE} Trans. Signal
  Process.}, vol.~45, no.~7, pp. 1800--1805, Jul. 1997.

\bibitem{wang2011frequency}
Z.~Wang, S.~Zhou, G.~B. Giannakis, C.~R. Berger, and J.~Huang,
  ``Frequency-domain oversampling for zero-padded {OFDM} in underwater acoustic
  communications,'' \emph{{IEEE} J. Ocean. Eng.}, vol.~37, no.~1, pp. 14--24,
  Jan. 2012.

\bibitem{su2008new}
B.~Su and P.~Vaidyanathan, ``New blind block synchronization for transceivers
  using redundant precoders,'' \emph{{IEEE} Trans. Signal Process.}, vol.~56,
  no.~12, pp. 5987--6002, Dec. 2008.

\bibitem{wang2006frames}
J.~Wang, J.~Song, Z.-X. Yang, L.~Yang, and J.~Wang, ``Frames theoretic analysis
  of zero-padding {OFDM} over deep fading wireless channels,'' \emph{{IEEE}
  Trans. Broadcast.}, vol.~52, no.~2, pp. 252--260, Jun. 2006.

\bibitem{van2012iterative}
D.~Van~Welden, H.~Steendam, and M.~Moeneclaey, ``Iterative decision-directed
  joint frequency offset and channel estimation for {KSP-OFDM},'' \emph{{IEEE}
  Trans. Commun.}, vol.~60, no.~10, pp. 3103--3110, Oct. 2012.

\bibitem{giannazpcp}
B.~Muquet, Z.~Wang, G.~B. Giannakis, M.~de~Courville, and P.~Duhamel, ``Cyclic
  prefixing or zero padding for wireless multicarrier transmissions?''
  \emph{{IEEE} Trans. Commun.}, vol.~50, no.~12, pp. 2136--2148, Dec. 2002.

\bibitem{li2008synchronization}
Y.~Li, H.~Minn, and R.~Rajatheva, ``Synchronization, channel estimation, and
  equalization in {MB-OFDM} systems,'' \emph{{IEEE} Trans. Wireless Commun.},
  vol.~7, no.~11, pp. 4341--4352, Nov. 2008.

\bibitem{chung2017preamble}
C.-D. Chung and W.-C. Chen, ``Preamble sequence design for spectral compactness
  and initial synchronization in {OFDM},'' \emph{{IEEE} Trans. Veh. Technol.},
  vol.~67, no.~2, pp. 1428--1443, Feb. 2018.

\bibitem{zhang2011autocorrelation}
J.~A. Zhang and X.~Huang, ``Autocorrelation based coarse timing with
  differential normalization,'' \emph{{IEEE} Trans. Wireless Commun.}, vol.~11,
  no.~2, pp. 526--530, Feb. 2011.

\bibitem{abdzadeh2012novel}
H.~Abdzadeh-Ziabari and M.~G. Shayesteh, ``A novel preamble-based frame timing
  estimator for {OFDM} systems,'' \emph{{IEEE} Commun. Lett.}, vol.~16, no.~7,
  pp. 1121--1124, July 2012.

\bibitem{sheng2010novel}
B.~Sheng, J.~Zheng, X.~You, and L.~Chen, ``{A} novel timing synchronization
  method for {OFDM} systems,'' \emph{{IEEE} Commun. Lett.}, vol.~14, no.~12,
  pp. 1110--1112, Dec. 2010.

\bibitem{LeNir2010}
\BIBentryALTinterwordspacing
V.~Le~Nir, T.~van Waterschoot, J.~Duplicy, and M.~Moonen, ``Blind coarse timing
  offset estimation for {CP-OFDM} and {ZP-OFDM} transmission over frequency
  selective channels,'' \emph{EURASIP Journal on Wireless Communications and
  Networking}, vol. 2009, no.~1, p. 262813, Jan. 2010. [Online]. Available:
  \url{https://doi.org/10.1155/2009/262813}
\BIBentrySTDinterwordspacing

\bibitem{bolcskei2001blind}
H.~Bolcskei, ``Blind estimation of symbol timing and carrier frequency offset
  in wireless {OFDM} systems,'' \emph{{IEEE} Trans. Commun.}, vol.~49, no.~6,
  pp. 988--999, Jun. 2001.

\bibitem{li2006orthogonal}
Y.~G. Li and G.~L. Stuber, \emph{Orthogonal frequency division multiplexing for
  wireless communications}.\hskip 1em plus 0.5em minus 0.4em\relax Springer
  Science, 2006.

\bibitem{hwang2008ofdm}
T.~Hwang, C.~Yang, G.~Wu, S.~Li, and G.~Y. Li, ``{OFDM} and its wireless
  applications: A survey,'' \emph{{IEEE} Trans. Veh. Technol.}, vol.~58, no.~4,
  pp. 1673--1694, May 2009.

\bibitem{Delay}
G.~Forney, ``Delay profile - an overview,'' \emph{{IEEE} Trans. Inf. Theory},
  vol.~18, no.~3, pp. 363--378, May 1972.

\bibitem{swindlehurst1998time}
A.~L. Swindlehurst, ``Time delay and spatial signature estimation using known
  asynchronous signals,'' \emph{IEEE Transactions on Signal Processing},
  vol.~46, no.~2, pp. 449--462, 1998.

\bibitem{roy1989esprit}
R.~Roy and T.~Kailath, ``Esprit-estimation of signal parameters via rotational
  invariance techniques,'' \emph{IEEE Transactions on acoustics, speech, and
  signal processing}, vol.~37, no.~7, pp. 984--995, 1989.

\bibitem{liu2014channel}
Y.~Liu, Z.~Tan, H.~Hu, L.~J. Cimini, and G.~Y. Li, ``Channel estimation for
  {OFDM},'' \emph{IEEE Communications Surveys \& Tutorials}, vol.~16, no.~4,
  pp. 1891--1908, 2014.

\bibitem{banelli2003theoretical}
P.~Banelli, ``Theoretical analysis and performance of {OFDM} signals in
  nonlinear fading channels,'' \emph{{IEEE} Trans. Wireless Commun.}, vol.~2,
  no.~2, pp. 284--293, Mar. 2003.

\bibitem{leon1994probability}
A.~Leon-Garcia, \emph{Probability and random processes for electrical
  engineering}.\hskip 1em plus 0.5em minus 0.4em\relax Pearson Education India,
  1994.

\bibitem{kay2013fundamentals}
S.~M. Kay, \emph{Fundamentals of statistical signal processing: Practical
  algorithm development}.\hskip 1em plus 0.5em minus 0.4em\relax Pearson
  Education, 2013, vol.~3.

\bibitem{press2007numerical}
W.~H. Press, S.~A. Teukolsky, W.~T. Vetterling, and B.~P. Flannery,
  \emph{Numerical recipes 3rd edition: The art of scientific computing}.\hskip
  1em plus 0.5em minus 0.4em\relax Cambridge university press, Sep. 2007.

\bibitem{K2020}
K.~{Pourtahmasi Roshandeh}, M.~{Mohammadkarimi}, and M.~{Ardakani}, ``{An
  Approximate Maximum Likelihood Time Synchronization Algorithm for Zero-padded
  OFDM in Channels with Impulsive Gaussian Noise},'' \emph{arXiv e-prints}, p.
  arXiv:2008.06586, Aug. 2020.

\bibitem{simon2007probability}
M.~K. Simon, \emph{Probability distributions involving Gaussian random
  variables: A handbook for engineers and scientists}.\hskip 1em plus 0.5em
  minus 0.4em\relax Springer Science \& Business Media, May 2007.

\bibitem{bibinger2013notes}
M.~Bibinger, ``Notes on the sum and maximum of independent exponentially
  distributed random variables with different scale parameters,'' \emph{arXiv
  preprint arXiv:1307.3945}, Jul. 2013.

\bibitem{gradshteyn2014table}
I.~S. Gradshteyn and I.~M. Ryzhik, \emph{Table of integrals, series, and
  products}.\hskip 1em plus 0.5em minus 0.4em\relax Academic press, May 2014.

\end{thebibliography}

\appendices

\section{Proof of Theorem \ref{uuonkml}}
\label{proof: theo uncorrelated}

Let us define sequences $\Set{S}$ and $\Set{Y}$ as follows
\begin{align}\label{khkkjh}
\Set{S}&\triangleq\big{\{}\dots,\tilde{s}[-1],\tilde{s}[0],\tilde{s}[1],\dots\big{\}} \\  \nonumber
&=\big{\{}\dots,0,0,{\bf{s}}_0^T,{\bf{s}}_1^T,\dots\big{\}},
\end{align}
and 
\begin{align}\label{khkkjh}
\Set{Y}&\triangleq 
\big{\{}\dots, \tilde{y}[-1], \tilde{y}[0],\tilde{y}[1],\dots\big{\}},  
\end{align}
where ${\bf{s}}_n$ is defined in \eqref{eq:11},  
$\tilde{s}[nn_{\rm{s}}+m]  \triangleq s_{n}[m]$,
and 
\begin{align}\label{yuiiir}
\tilde{y}[u]=\sum_{l=0}^{n_{\rm{h}}-1} h[l] \tilde{s}[u-l]+{w}[u].
\end{align}
with  ${w}[u]\sim \mathcal{CN}(0,\sigma^2_{\rm{w}})$ as \ac{awgn}

Since $\Set{S}$ is composed of OFDM samples  $x_n(mT_{\rm{sa}})$, which are modeled as zero-mean \ac{iid} complex Gaussian random variables, and zero-padded elements, 
we have
\begin{align}
\mathbb{E}\{\tilde{y}[u]\}=0, 
\end{align}
and
\begin{align}
 \mathbb{E}\{\tilde{s}[v]\tilde{s}^*[u]\}={0},\,\,\,\,\,\  v \neq u. 
 \end{align}
 For $u=v$, we either have  $\mathbb{E}\{\tilde{s}[v]\tilde{s}^*[u]\}={0}$ or $\mathbb{E}\{\tilde{s}[v]\tilde{s}^*[u]\}= \sigma_{\rm{x}}^2$, where the former is valid for the zero-padded samples, and the latter is obtained from \eqref{909m0p} for OFDM samples.
Since the noise component $w[n]$ in \eqref{yuiiir} is zero-mean \ac{awgn}, we can write\footnote{It is obvious that for the observation samples containing noise only samples, $\mathbb{E}\big{\{}{\tilde{y}[u]}\tilde{y}^*[v]\big{\}}=0$, $v \neq u$.}
\begin{align}\label{0123df}
\mathbb{E}\big{\{}{\tilde{y}[u]}\tilde{y}^*[v]\big{\}} =\sigma_{\rm{x}}^2\sum\limits_{l = 0}^{{n_{\rm{h}}} - 1}  \mathbb{E}\big{\{}{h[l]{h^*[v - u + l]}}\big{\}},  
\end{align}
where $u \neq v$.
Because the channel taps in \eqref{7u8i0000} are uncorrelated random variables, we can write  
\begin{align}\label{uiomoer}
\mathbb{E}\big{\{}{\tilde{y}[u]}\tilde{y}^*[v]\big{\}}=0,\,\,\,\,\ u \neq v,  \end{align}
which implies that $\tilde{y}[u]$ and  $\tilde{y}[v]$ are also uncorrelated random variables.

Finally,  we can easily 
conclude that the elements of $\bf{y}$ given hypothesis ${\rm{H}}_d$ 
share the same property with $\Set{Y}$ in the context of correlation 
since $\bf{y}$ is obtained by windowing $\Set{Y}$.

\section{ Proof of Theorem \ref{theo: pdf y}}
\label{proof: theo}

Let us write the in-phase component of the received sample in \eqref{uiomoer21} as follows 
\begin{align}\label{89989090}
y_{n_{\rm{I}}}[m]=v_{n_{\rm{I}}}[m]+w_{n_{\rm{I}}}[m],
\end{align}
where $y_{n_{\rm{I}}}[m] \triangleq\Re\{y_n{[m]}\}$, 
$w_{n_{\rm{I}}}[m] \triangleq \Re\{w_n{[m]}\}$, and \begin{align}\nonumber
v_{n_{\rm{I}}}[m] \triangleq\Re\{v_n{[m]}\}=\Re\Bigg{\{}\sum_{l=0}^{n_{\rm{h}}-1} h[l] s_{n}[m-l]\Bigg{\}}. 
\end{align}
We consider that 
$v_{n_{\rm{I}}}[m]$ and $w_{n_{\rm{I}}}[m]$ are the realizations of the random variables $V_{n_{\rm{I}}}[m]$ and $W_{n_{\rm{I}}}[m]$, respectively. 
We also denote the \ac{pdf} of $V_{n_{\rm{I}}}[m]$ and $W_{n_{\rm{I}}}[m]$ with $f_{V_{n_{\rm{I}}}[m]}(v | {\rm{H}}_0)$ and $f_{W_{n_{\rm{I}}}[m]}(w | {\rm{H}}_0)$.

{\it Case 1}: $m \in  \{ m ~|~ 0 \le m \le n_{\rm{s}}-1  ~{\rm{when}}~  n<0 \} \cup \{ m ~|~ n_{\rm{x}}+n_{\rm{h}}-1 \le m \le n_{\rm{s}}-1 ~{\rm{when}}~  n \ge 0\}$

\begin{figure*}[t] 
\label{eq:6767}

\setcounter{equation}{43}
 \begin{equation} \label{eq: expa conv 2}
v_{n_{\rm{I}}}[m] \hspace{-0.2em}= \hspace{-0.2em}
     \begin{cases}
      \sum^{m}_{l=0} h_{\rm{I}}[l] x_{n_{\rm{I}}}[m-l] - h_{\rm{Q}}[l] x_{n_{\rm{Q}}}[m-l]  ~~~~~~~~~~~~~~~~~~~~~~~~ 0 \le m \le n_{\rm{h}}-2,\\
       \sum^{n_{\rm{h}}-1}_{l=0} h_{\rm{I}}[l] x_{n_{\rm{I}}}[m-l] - h_{\rm{Q}}[l] x_{n_{\rm{Q}}}[m-l]  ~~~~~~~~~~~~~~~~~~~~~~ n_{\rm{h}}-1 \le m \le n_{\rm{x}}-1,\\
       \sum^{n_{\rm{h}}-1}_{l=m-n_{\rm{x}}+1} h_{\rm{I}}[l] x_{n_{\rm{I}}}[m-l] - h_{\rm{Q}}[l] x_{n_{\rm{Q}}}[m-l]  ~~~~~~~~~~~~~~~~ n_{\rm{x}} \le m \le n_{\rm{x}}+n_{\rm{h}}-2,\\
       0 ~~~~~~~~~~~~~~~~~~~~~~~~~~~~~~~~~~~~~~~~~~~~~~~~~~~~~~~~~~~~~~~~~~~~~~ n_{\rm{x}}+n_{\rm{h}}-1 \le m \le n_{\rm{s}}-1
     \end{cases}
\end{equation}
\vspace{-1.2pt}
\end{figure*}


Let us write  $v_{n_{\rm{I}}}[m]$ as equation \eqref{eq: expa conv 2}, at the top of the next page, where $h_{\rm{I}}[l]\triangleq \Re\{h[l]\}$, $h_{\rm{Q}}[l]\triangleq \Im\{h[l]\}$, 
 $x_{n_{\rm{I}}}[m]\triangleq \Re\{x_n[m]\}$, and  $x_{n_{\rm{Q}}}[m]\triangleq \Im\{x_n[m]\}$. 
By replacing \eqref{eq: expa conv 2} into \eqref{89989090}, we can write (see Fig. \ref{fig: conv})
\begin{align}\label{eq:47}
y_{n_{\rm{I}}}[m]=w_{n_{\rm{I}}}[m],\,\,\,\,\,\,
n_{\rm{x}}+n_{\rm{h}}-1 \le m \le n_{\rm{s}}-1. 
\end{align}
Also, given ${\rm{H}}_0$, we have 
\begin{align}\label{eq:48}
y_{n_{\rm{I}}}[m]=w_{n_{\rm{I}}}[m],\,\,\,\,\,\, n<0. 
\end{align}
By using $W_{n_{\rm{I}}}[m] \sim f_{W_{n_{\rm{I}}}[m]}(w | {\rm{H}}_0)=
\mathcal{CN}(0,\sigma^2_{\rm{w}}/2)$, \eqref{eq:47}, and \eqref{eq:48}, we obtain \eqref{uioooppioio}. 

 {\it Case 2}: $m \in  \{ m ~|~ 0 \le m \le n_{{\rm{x}}}+n_{{\rm{h}}}-2~{\rm{when}}~  n \ge 0\}$

Since $W_{n_{\rm{I}}}[m]$ and $V_{n_{\rm{I}}}[m]$ are independent random variables, we can write 
the \ac{pdf} of $y_{n_{\rm{I}}}[m]$ in \eqref{89989090} given hypothesis ${\rm{H}}_0$ as the convolution of their \ac{pdf}s. Prior to convolution derivation, we first need to derive the \ac{pdf} of $V_{n_{\rm{I}}}[m]$ for $ 0 \le m\le n_{\rm{x}}+n_{\rm{h}}-2$. To obtain the \ac{pdf} of  $V_{n_{\rm{I}}}[m]$, we can employ the \ac{chf} method.  


By using (6.1) in \cite{simon2007probability} and  \eqref{eq: expa conv 2}, we can write the CHF of $V_{n_{\rm{I}}}[m]$ given ${\rm H}_0$ as follows 
\begin{align} \label{eq: mgf gen}
\phi_{V_{n_{\rm{I}}}[m]|{\rm{H}}_0} (t)  & \triangleq \frac{1}{\prod^{b}_{k=a}} \Big{(} 1+\frac{\sigma_{h_k}^2 \sigma_{{\rm{x}}}^2 t^2}{4} \Big{)}\\ \nonumber  
& =\prod^{b}_{k=a} \frac{1}{\Big( 1+ j \frac{\sigma_{h_k} \sigma_{{\rm{x}}} }{2}t \Big) } \frac{1}{ \Big(1- j \frac{\sigma_{h_k} \sigma_{{\rm{x}}} }{2}t\Big)}
\end{align}
 for $ 0 \le m\le n_{{\rm{x}}}+n_{{\rm{h}}}-2$, where 
$(a,b)$ is given in \eqref{eq: a b no isi}.

The \ac{chf} of the random variable $X \triangleq X_1+X_2+ \dots +X_L$, where $X_i$ and $X_j$ are independent random variables, is given as follows
\begin{align}\label{8989090}
\phi_X(t)=\phi_{X_1}(t)\phi_{X_2}(t)\dots\phi_{X_L}(t). 
\end{align}
By employing \eqref{8989090}, we can write the random variable $V_{n_{\rm{I}}}[m]$ with the \ac{chf} in \eqref{eq: mgf gen} as the summation of independent random variables as follows
\begin{align}\label{eq: decomposee}
    V_{n_{\rm{I}}}[m] =  \sum^{b}_{k=a} (E_{k}-E'_{k})
    =
     V_1 - V_2,
\end{align}
\noindent where  $V_1 \triangleq \sum^{b}_{k=a} E_{k}$,  $V_2 \triangleq \sum^{b}_{k=a} E'_{k}$, and $(a,b)$ is given in \eqref{eq: a b no isi}. In \eqref{eq: decomposee}, $E_k$ and $E'_k$ are independent and exponentially distributed random variables with rate parameter $\lambda_k=(\sigma_{{\rm{h}}_k} \sigma_{\rm{x}}/2)^{-1}$. Using equation (7) in \cite{bibinger2013notes}, we can write the \ac{pdf} of $V_1$  as follows
\begin{equation}
f_{V_1}(v_1) =  \prod_{i=a}^{b} \lambda_i \sum_{j=a}^{b} \frac{ e^{-\lambda_j v_1} }{ \prod_{k=a, k \neq j}^{b} (\lambda_k - \lambda_j) }.
\end{equation}
Similar expression holds for the \ac{pdf} of $V_2$.
Since $V_{n_{\rm{I}}}[m]=V_1-V_2$, and $V_1, V_2 \in [0,\infty)$, then, $V_{n_{\rm{I}}}[m] \in (-\infty, \infty)$. 

The PDF of the sum of two independent random variables is the convolution of their PDFs. Since $V_1$ and $V_2$ are independent random variables,  
for $v \ge 0$, we can write 
\vspace{-0.1 cm}
\begin{align}\label{0909opop}
\hspace{-1em}f_{V_{n_{\rm{I}}}[m]}(v | \rm{H}_0)& =  \int_{0}^{\infty} f_{V_1}(v+v_2) f_{V_2}(v_2) dv_2 \\ \nonumber 
&= \int_{0}^{\infty}
\prod_{i=a}^{b} \lambda_i \sum_{j=a}^{b} \frac{ e^{-\lambda_j (v+v_2)} }{ \prod_{k=a, k \neq j}^{b} (\lambda_k - \lambda_j)} \\ \nonumber
&~~~~~\times \prod_{r=a}^{b} \lambda_r \sum_{n=a}^{b} \frac{ e^{-\lambda_n v_2} }{ \prod_{p=a, p \neq j}^{b} (\lambda_p - \lambda_n)} dv_2
 \\ \nonumber
&=  \Bigg(\prod_{i=a}^{b} \lambda_i \Bigg)^2 \sum_{j=a}^{b}  \sum_{n=a}^{b} \frac{ 1}{ \prod_{k=a, k \neq j}^{b} (\lambda_k - \lambda_j)} \\ \nonumber
&~~~~~ \times\frac{ 1 }{ \prod_{p=a, p \neq j}^{b} (\lambda_p - \lambda_n)} \frac{e^{-\lambda_j v}}{\lambda_j + \lambda_n}.
\end{align}
\vspace{-0.1cm}
Similarly, for $v \le 0$, we obtain
\begin{align} \nonumber
\hspace{-0.2em}f_{V_{n_{\rm{I}}}[m]}(v | \rm{H}_0)& =    \Bigg( \prod_{i=a}^{b} \lambda_i \Bigg)^2 \sum_{j=a}^{b}  \sum_{n=a}^{b} \frac{ 1}{ \prod_{k=a, k \neq j}^{b} (\lambda_k - \lambda_j)} \\ \label{pdf_v}
&~ \times\frac{ 1 }{ \prod_{p=a, p \neq j}^{b} (\lambda_p - \lambda_n)} \frac{e^{\lambda_j v}}{\lambda_j + \lambda_n}.
\end{align}

Now, we can write the \ac{pdf} of the received sample $y_{n_{\rm{I}}}[m]$ as the convolution of $f_{V_{n_{\rm{I}}}[m]}(v | {\rm{H}}_0)$ and $f_{W_{n_{\rm{I}}}[m]}(w | {\rm{H}}_0)$ as follows 
\begin{equation} \label{eq: C1+C2}
\begin{split}
f_{Y_{n_{\rm{I}}}[m]}(y & | {\rm{H}}_0) = \int_{-\infty}^{\infty} f_{W_{n_{\rm{I}}}[m]}(y-v) f_{V_{n_{\rm{I}}}[m]}(v | {\rm{H}}_0) dv \\
&= \underbrace{ \int_{0}^{\infty} f_{W_{n_{\rm{I}}}[m]}(y-v) f_{V_{n_{\rm{I}}}[m]}(v | {\rm{H}}_0) dv}_{C1} \\ &~~~+ \underbrace{ \int_{-\infty}^{0} f_{W_{n_{\rm{I}}}[m]}(y-v) f_{V_{n_{\rm{I}}}[m]}(v | {\rm{H}}_0) dv}_{C2}.
\end{split}
\end{equation}
By using \eqref{0909opop} and  $f_{W_{n_{\rm{I}}}[m]}(w | {\rm{H}}_0)=
\mathcal{CN}(0,\sigma^2_{\rm{w}}/2)$, the first integral in \eqref{eq: C1+C2} can be obtained  as follows
\vspace{-0.1cm}
 \begin{equation} \label{eq: C1}
\begin{split}
 C_1 &= \Bigg(\prod_{i=a}^{b} \lambda_i \Bigg)^2 \sum_{j=a}^{b}  \sum_{n=a}^{b} \frac{ 1}{ \prod_{k=a, k \neq j}^{b} (\lambda_k - \lambda_j)} \frac{1}{(\lambda_j + \lambda_n) \pi \sigma_{\rm w}^2} \\
&~~~ \times \frac{ 1 }{ \prod_{p=a, p \neq j}^{b} (\lambda_p - \lambda_n)}   e^{-\frac{y^2}{\sigma_{\rm w}^2}} \int_{0}^{\infty} e^{-\frac{v^2}{\sigma^2_{\rm{w}}}-(\lambda_j- \frac{2y}{\sigma^2_{\rm{w}}})v} dv    \\
& \stackrel{(g)}{=}  \Bigg( \prod_{i=a}^{b} \lambda_i \Bigg)^2 \sum_{j=a}^{b}  \sum_{n=a}^{b} \frac{ e^{(\frac{\lambda_j \sigma_{\rm{w}}}{2})^2} }{ \prod_{k=a, k \neq j}^{b} (\lambda_k - \lambda_j)} \frac{1}{2(\lambda_j + \lambda_n) } \\
&~~~ \times \frac{ 1 }{ \prod_{p=a, p \neq j}^{b} (\lambda_p - \lambda_n)}    e^{-\lambda_j y} \Bigg( 1-\Phi \Big(\frac{\lambda_j \sigma_{\rm{w}}}{2}-\frac{y}{\sigma_{\rm{w}}}\Big) \Bigg),
\end{split}
\end{equation}
\noindent where (g) comes from \cite{gradshteyn2014table} (page 336, 3.322, formula 2), and $\Phi(x)= {\rm{erf}}(x)= \frac{2}{\sqrt{\pi}}\int_{0}^{x} e^{-t^2} dt$ denotes the Gaussian error function. Analogous to \eqref{eq: C1}, by using \eqref{pdf_v} and after some mathematical simplifications, we obtain
\begin{equation} \label{eq: C2}
\begin{split}
C_2 &=  \Bigg( \prod_{i=a}^{b} \lambda_i \Bigg)^2 \sum_{j=a}^{b}  \sum_{n=a}^{b} \frac{ 1}{ \prod_{k=a, k \neq j}^{b} (\lambda_k - \lambda_j)} \\
&~~~ \times\frac{ 1 }{ \prod_{p=a, p \neq j}^{b} (\lambda_p - \lambda_n)} \frac{1}{2(\lambda_j + \lambda_n) } \\ 
&~~~ \times e^{\Big{(}\frac{\lambda_j \sigma_{\rm{w}}}{2}\Big{)}^2} e^{\lambda_j y} \Bigg(1-\Phi \Big(\frac{\lambda_j \sigma_{\rm{w}}}{2}+\frac{y}{\sigma_{\rm{w}}} \Big) \Bigg).
\end{split}
\end{equation}
Finally, by substituting \eqref{eq: C1} and \eqref{eq: C2} into \eqref{eq: C1+C2}, \eqref{eq: pdf} is derived. 
One can  derive an identical expression for the quadrature component of the received samples by following the same procedure.

\section{}
\label{proof: theo final}

In order to perceive the relation between 
$f_{Y_{n}[m]}(y|{\rm{H}}_0)$ and
$ f_{Y_{n}[m]}(y|{\rm{H}}_d)$, we rely on the following  observations resulting from Theorem \ref{theo: pdf y}.

1) 
For $n  \ge 0$,
there is a repeating pattern in the \ac{pdf} of the received samples due to
 the zero-padded guard interval and Gaussianity of the OFDM samples, i.e., $x_n(mT_{\rm{sa}}) \sim \mathcal{CN}(0,\sigma^2_{\rm{x}})$. Hence, we have  
\begin{align}\label{21230o}
 f_{Y_{(n+q)_{\rm{I}}}[m]}(y_{\rm{I}}|{\rm{H}}_0)= 
f_{Y_{n_{\rm{I}}}[m]} (y_{\rm{I}}|{\rm{H}}_0), \,\,\,\,\ q \ge 0, 
\end{align}
The same equality holds for $f_{Y_{n_{\rm{Q}}}[m]} (y_{\rm{Q}}|{\rm{H}}_0)$. This repetition pattern is shown in Fig. \ref{fig: conv}. 
For simplicity of presentation, we remove the index of the OFDM symbol due to this periodicity and write
\begin{align}\label{8uip}
f_{Y_{n_{\rm{I}}}[m]} (y_{\rm{I}}|{\rm{H}}_0) \triangleq
{f}_{Y_{{\rm{I}}}[m]} (y_{\rm{I}}|{\rm{H}}_0),\,\,\,\,\,\ n  \ge 0.
\end{align}
The same definition holds for $f_{Y_{n_{\rm{Q}}}[m]} (y_{\rm{Q}}|{\rm{H}}_0)$. Following the notation in \eqref{8uip}, we can write 
\begin{align} \label{ioipop}
 f_{Y_n[m]}( y | {\rm{H}}_0 )  
 & \approx f_{Y_{n_{\rm{I}}}[m])} ( y_{\rm{I}} | {\rm{H}}_0 ) f_{Y_{n_{\rm{Q}} }[m]} ( y_{\rm{Q}} | {\rm{H}}_0 ) \\ \nonumber 
 &\triangleq {f}_{Y_{\rm{I}}[m]} ( y_{\rm{I}} | {\rm{H}}_0 ) {f}_{Y_{\rm{Q}}[m]} ( y_{\rm{Q}} | {\rm{H}}_0 ) 
 \triangleq \tilde{f}_{Y[m]}( y| {\rm{H}}_0).
\end{align}

\begin{figure*}
\vspace{-2em}
\centering 
  \includegraphics[scale=0.7]{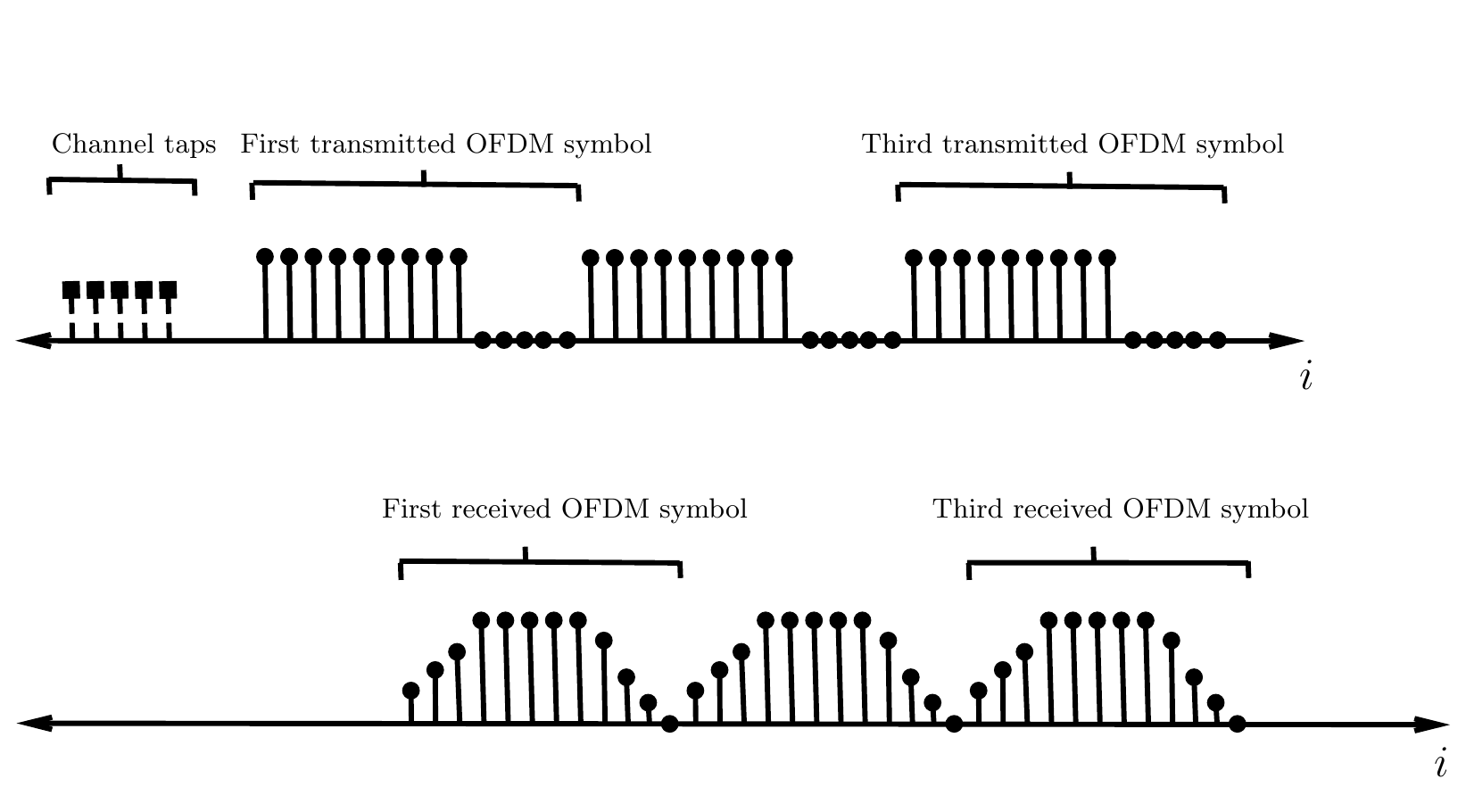}
  \vspace{-1em}
  \caption{Convolution of OFDM symbols with a multi-tap channel results in four region given in \eqref{eq: expa conv 2}.}\label{fig: conv}
\end{figure*}

2) 
For $n<0$, $f_{Y_{n}[m]} (y|{\rm{H}}_0)$ is the \ac{pdf} of the complex Gaussian noise. For simplicity of presentation, we define 
\begin{align}\label{8uipuiui}
\nonumber f_{Y_{n}[m]} (y|{\rm{H}}_0) & \approx f_{Y_{n_{\rm{I}}}[m]}(y_{\rm{I}}|{\rm{H}}_0) f_{Y_{n_{\rm{Q}}}[m]}(y_{\rm{Q}}|{\rm{H}}_0)\\ 
&\triangleq {f}_{Y_{\rm{I}}[-]} \big( y_{\rm{I}} | {\rm{H}}_0 \big) {f}_{Y_{\rm{Q}}[-]} \big( y_{\rm{Q}} | {\rm{H}}_0 \big)\\ \nonumber
&\triangleq \tilde{f}_{Y[-]} (y|{\rm{H}}_0),\,\,\,\,\,\,\,\,\,\,\,\,\,\,\,\,\,\,\,\,\,\,\,\,\,\,\,\,\  n<0.
\end{align}

By using \eqref{ioipop} and \eqref{8uipuiui}, we can define the vector of \ac{pdf}  for the observations vectors $\{\dots, {\bf{y}}_{-2}^{\rm{T}},{\bf{y}}_{-1}^{\rm{T}},{\bf{y}}_0^{\rm{T}},{\bf{y}}_1^{\rm{T}},{\bf{y}}_2^{\rm{T}},\dots\}$ given hypothesis ${\rm{H}}_0$
 as follows

 ${\bf f}_{\bf{Y}}({\bf \cdot}; {\rm{H}}_0) \triangleq$
\begin{equation}\label{eq: pdf matrix H0} 
\hspace{-0.5cm}   \left[\begin{array}{@{}c}
 \vspace{0.2cm}  
 \vdots\\
 f_{Y_{-1}[n_{{\rm{s}}}-2]}\big(\cdot| {\rm{H}}_0 \big) \\ \vspace{0.2cm} f_{Y_{-1}[n_{{\rm{s}}}-1]}\big(\cdot| {\rm{H}}_0 \big)  \\   \hdashline  \\
 \vspace{0.2cm} f_{Y_0[0]}\big(\cdot| {\rm{H}}_0 \big) \\  f_{Y_{0}[1]}\big(\cdot | {\rm{H}}_0\big)  \\ \vdots \\ \vspace{0.2cm} f_{Y_0[n_{\rm{s}}-1]}\big(\cdot | {\rm{H}}_0 \big)  \\   \hdashline  \\
 \vspace{0.2cm} f_{Y_1[0]}\big(\cdot | {\rm{H}}_0 \big) \\  f_{Y_1[1]}\big(\cdot | {\rm{H}}_0 \big)  \\ \vdots \\ \vspace{0.2cm} f_{Y_1[n_{\rm{s}}-1]}\big(\cdot | {\rm{H}}_0 \big)   \\      \hdashline \vspace{-0.3cm} \\
\vdots \vspace{0.3cm}
 \end{array}\right] = 
\hspace{-0.1cm} 
   \left[\begin{array}{@{}c@{}}
  \vdots \\ 
 \tilde{f}_{Y[-]}\big(\cdot| \rm{H}_0 \big) \\ \vspace{0.2cm} \tilde{f}_{Y[-]}\big(\cdot| \rm{H}_0 \big)  \\   \hdashline  \\
 \vspace{0.2cm} \tilde{f}_{Y[0]}\big(\cdot| \rm{H}_0 \big) \\  \tilde{f}_{Y[1]}\big(\cdot | \rm{H}_0\big)  \\ \vdots \\ \vspace{0.2cm} \tilde{f}_{Y[n_{\rm{s}}-1]}\big(\cdot | \rm{H}_0 \big)  \\   \hdashline  \\
 \vspace{0.2cm} \tilde{f}_{Y[0]}\big(\cdot | \rm{H}_0 \big) \\  \tilde{f}_{Y[1]}\big(\cdot | \rm{H}_0 \big)  \\ \vdots \\ \vspace{0.2cm} \tilde{f}_{Y[n_{\rm{s}}-1]}\big(\cdot | \rm{H}_0 \big)   \\      \hdashline \vspace{-0.3cm} \\
\vdots \vspace{0.3cm}
\end{array}
\right].
\end{equation}
As seen in \eqref{eq: pdf matrix H0}, the right-hand side vector simply represents the repetition pattern of the \ac{pdf}. 

The vector of \ac{pdf} for the observation vector
${\bf{y}}$ in \eqref{eq:14} given hypothesis ${\rm{H}}_d$ is expressed as follows
\begin{align}\label{uiuiopl}
{\bf f}_{\bf{Y}}({\bf{y}} | {\rm{H}}_d)={\bf f}^{(d:d+Nn_{\rm{s}}-1)}_{\bf{Y}} ({\bf y}; {\rm{H}}_0 ),
\end{align}
where ${{\bf f}^{(d:d+Nn_{\rm{s}}-1)}_{\bf{Y}}} (\cdot ; {\rm{H}_0})$ is defined in \eqref{6yo09}. 
By using
\eqref{eq:29} and \eqref{eq:30},   we can write ${\bf f}_{\bf{Y}}({\bf{y}} | {\rm{H}}_d)$ in \eqref{uiuiopl} based on $\tilde{f}_{Y[-]} (\cdot|{\rm{H}}_0)$ and $\tilde{f}_{Y[m]}( \cdot| {\rm{H}}_0)$, $m=0,1,\dots,n_{\rm{s}}-1$, as shown in \eqref{dpos} and \eqref{dpos1}. 
It is worth mentioning that due to the repetition pattern in ${\bf f}_Y({\bf \cdot}; {\rm{H}}_0)$, the truncated \ac{pdf} vector in \eqref{uiuiopl} contains $N-1$ full blocks of $\big{[}\tilde{f}_{Y[0]}(\cdot | {\rm{H}}_0 )
, \tilde{f}_{Y[1]}(\cdot | {\rm{H}}_0 ), \ \dots, \ 
\tilde{f}_{Y[n_{\rm{s}}-1]}(\cdot | {\rm{H}}_0 )\big{]}^T$.

\end{document}